\documentclass{article}

\usepackage{amsmath,amssymb,amsthm,bbm,natbib,color,hyperref,nccmath}
\usepackage[margin=1in]{geometry}
\usepackage{caption}
\usepackage{graphicx}
\usepackage{subcaption}
\usepackage{afterpage}
\usepackage[hang,flushmargin]{footmisc}
\usepackage{float}
\usepackage{enumitem,kantlipsum}
\usepackage{multirow}
\usepackage{algorithm}
\usepackage{makecell}

\newtheorem{theorem}{Theorem}
\newtheorem{lemma}{Lemma}
\newtheorem{proposition}{Proposition}
\newtheorem{assumption}{Assumption}
\newtheorem{remark}{Remark}

\newtheorem{corollary}{Corollary}

\newcommand{\eps}{\epsilon}
\newcommand{\EE}[1]{\mathbb{E}\left[{#1}\right]}
\newcommand{\EEst}[2]{\mathbb{E}\left[{#1}\  \middle| \ {#2}\right]}
\newcommand{\Ep}[2]{\mathbb{E}_{{#1}}\left[{#2}\right]}
\newcommand{\Epe}[3]{\mathbb{E}_{{#1}}^{{#2}}\left[{#3}\right]}
\newcommand{\Epst}[3]{\mathbb{E}_{{#1}}\left[{#2}\  \middle| \ {#3}\right]}
\newcommand{\PP}[1]{\mathbb{P}\left\{{#1}\right\}}
\newcommand{\PPst}[2]{\mathbb{P}\left\{{#1}\  \middle| \ {#2}\right\}}
\newcommand{\Ppst}[3]{\mathbb{P}_{{#1}}\left\{{#2}\  \middle| \ {#3}\right\}}
\newcommand{\Pp}[2]{\mathbb{P}_{{#1}}\left\{{#2}\right\}}
\newcommand{\One}[1]{{\mathbbm{1}}\left\{{#1}\right\}}
\newcommand{\pt}{\frac{\partial}{\partial t}}
\newcommand{\X}{\mathcal{X}}
\newcommand{\Y}{\mathcal{Y}}
\newcommand{\R}{\mathbb{R}}
\newcommand{\N}{\mathbb{N}}
\newcommand{\indep}{\perp \!\!\! \perp}

\newcommand{\iidsim}{\stackrel{\textnormal{i.i.d.}}{\sim}}
\allowdisplaybreaks
\newcommand{\IF}{\mathrm{IF}}
\newcommand{\T}{\mathcal{T}}
\newcommand{\J}{\mathcal{J}}

\newcommand{\chb}{\widehat{C}_{\beta^*}}
\newcommand{\chh}{\widehat{C}_h}

\newcommand{\my}[1]{\textcolor{cyan}{[MY: #1]}}

\begin{document}
\title{Inference on Nonlinear Counterfactual Functionals under a Multiplicative IV Model}
\date{\today}

\author{
Yonghoon Lee\textsuperscript{1}\thanks{yhoony31@wharton.upenn.edu} \,
Mengxin Yu\textsuperscript{2}\thanks{myu@wustl.edu} \,
Jiewen Liu\textsuperscript{3}\thanks{jiewen.liu@pennmedicine.upenn.edu} \,
Chan Park\textsuperscript{4}\thanks{parkchan@illinois.edu} \\
Yunshu Zhang\textsuperscript{3}\thanks{yunshu.zhang@pennmedicine.upenn.edu} \,
James M. Robins\textsuperscript{5}\thanks{robins@hsph.harvard.edu} \,
Eric J. Tchetgen Tchetgen\textsuperscript{1,3}\thanks{
Address for correspondence: Eric J. Tchetgen Tchetgen, 407 Academic Research Building, 265 South 37th Street, Philadelphia, PA 19104. Email: ett@wharton.upenn.edu}
}

\date{
\vspace{1em}
\begin{minipage}{0.95\textwidth}
\centering
\small
\textsuperscript{1} Department of Statistics and Data Science, Wharton School, University of Pennsylvania \\
\textsuperscript{2} Department of Statistics and Data Science, Washington University in St. Louis \\
\textsuperscript{3} Department of Biostatistics, Perelman School of Medicine, University of Pennsylvania \\
\textsuperscript{4} Department of Statistics, University of Illinois Urbana-Champaign \\
\textsuperscript{5} Department of Biostatistics, Harvard T.H. Chan School of Public Health
\end{minipage}
}

\maketitle

\begin{abstract}

Instrumental variable (IV) methods play a central role in causal inference, particularly in settings where treatment assignment is confounded by unobserved variables. IV methods have been extensively developed in recent years and applied across diverse domains, from economics to epidemiology. In this work, we study the recently introduced multiplicative IV (MIV) model and demonstrate its utility for causal inference beyond the average treatment effect. In particular, we show that it enables identification and inference for a broad class of counterfactual functionals characterized by moment equations. This includes, for example, inference on quantile treatment effects. We develop methods for efficient and multiply robust estimation of such functionals, and provide inference procedures with asymptotic validity. Experimental results demonstrate that the proposed procedure performs well even with moderate sample sizes.
    
\end{abstract}


\section{Introduction}

In the study of causal effects, a fundamental challenge is dealing with unmeasured confounding. Traditional methods for addressing this issue typically involve randomizing treatment assignment or assuming that all confounders are observed. However, these approaches are often not feasible or realistic in practice. An alternative approach is to use a so-called instrumental variable (IV), which is apriori known to affect the outcome only through its effect on the treatment and not otherwise, and  is also known to be independent of all unmeasured confounders.

For example, consider a double blind placebo-controlled randomized trial where one observes data from a sample of patients, some of whom are randomly assigned to take a new drug against a certain disease outcome, e.g., anti-high blood pressure medication, or anti-obesity GLP-1-based drug regimen, versus a placebo control, with the outcome of interest defined as disease progression or resolution measured at the end of a fixed follow-up period. As common in many randomized experiments, a subset of participants may choose not to take their assigned treatment for unknown reasons, in which case unmeasured confounding of treatment uptake cannot be ruled out with certainty. In such a setting, the randomized treatment assignment is a well-known candidate as an instrumental variable. This is because it is determined randomly by design, and therefore independent of all confounding factors whether observed or unmeasured; and provided double-blinding is successful such that other protocol violations can be avoided, treatment assignment can only impact the outcome through treatment uptake. By leveraging instrumental variables in such settings, one can make inferences about the causal effect of treatment taken on the outcome, even in the presence of unmeasured confounding, potentially leading to more reliable and unbiased causal inferences.

While a valid instrument can in principle provide a valid test of a null causal effect, or partial identification of a causal effect, point identification typically requires imposing an additional condition.  For this reason, the causal inference literature across several disciplines has increasingly focused on learning with instrumental variables, by considering various estimands corresponding to different underlying additional instrumental variable assumption imposed for identification. These targets of inference include the average treatment effect (ATE)~\citep{wright1928tariff,goldberger1972structural}, and ~\citep{wang2018bounded} local average treatment effect (LATE)~\citep{angrist1996identification}, the average treatment effect on the treated (ATT)~\citep{robins1994correcting}, and others, under assumptions that range from parametric structural linear equations models, to nonparametric models that evoke certain monotonicity or homogeneity conditions.

In this work, we consider the so-called \emph{multiplicative instrumental variable model} (MIV) recently introduced by~\cite{liu2025multiplicativeinstrumentalvariablemodel}, which posits a proportional propensity score model for treatment as a function of the instrument and unmeasured confounders. While~\cite{liu2025multiplicativeinstrumentalvariablemodel} establish identification of the ATT under the MIV model, here we show that the MIV in fact provides identification and inference for a broader range of causal estimands among the treated---specifically, any functional of the counterfactual distribution which can be expressed as the unique solution to a corresponding moment equation. We present theoretical results for identification, semiparametric efficient estimation, and inference of the target functional, ensuring asymptotic validity.

\subsection{Notations}
We denote the real space by $\R$ and the space of positive real numbers as $\R^+$. We write $\N$ to denote the set of positive integers. For a sequence of random vectors $(X_n)$ and a sequence of positive real numbers $(\tau_n)$, we write $X_n = o_P(\tau_n)$ to denote that $\lim_{n \rightarrow \infty} \PP{\|\tfrac{X_n}{\tau_n}\|> \eps} = 0$ holds for any $\eps > 0$, i.e., $\|X_n/\tau_n\|$ converges to 0 in probability. For a random sequence $(X_n)$, $X_n \xrightarrow{D} P$ denotes that $X_n$ converges to $P$ in distribution. For random variables $X$ and $Y$, $F_X$ denotes the cumulative distribution function (cdf) of $X$, and $F_{X \mid Y}$ denotes the conditional cdf of $X$ given $Y$. For a positive integer $n$, $[n]$ denotes the set $\{1,2,\cdots,n\}$ and $(a_i)_{i \in [n]}$ denotes the vector $(a_1,\cdots,a_n)$. For $\alpha \in (0,1)$, $z_\alpha$ denotes the $(1-\alpha)$-quantile of the standard normal distribution.

\subsection{Problem setup}

We consider a standard setting of data with instrumental variable, where we observe  $O_1, O_2, \cdots, O_n \iidsim P$, where $O_i = (X_i,Z_i,A_i,Y_i)$. Here, $X \in \X$ denotes baseline covariates, $Z \in \{0,1\}$ is the instrumental variable, $A \in \{0,1\}$ is the treatment, and $Y \in \Y \subset \R$ is the outcome. We denote the vector of unmeasured confounders by $U \in \mathcal{U}$, and write $Y^{a=0}$ and $Y^{a=1}$ to denote the counterfactual outcomes when setting $A=0$ and $A=1$, respectively. We suppose that the following core IV assumptions hold. 

\begin{assumption}\label{asm:iv}
    The random variables $(X,Z,A,Y,U)$ satisfy
    \begin{enumerate}
        \item $Y=Y^{a=0}(1-A) + Y^{a=1} A$ almost surely, (consistency)
        \item $(A,Z) \indep Y^{a=0} \mid X, U$, (weak ignorability and exclusion-restriction)
        \item $U \indep Z \mid X$. (IV independence)
    \end{enumerate}
\end{assumption}
Note that when one can reasonably conceive of an intervention on the instrument, the second condition in Assumption~\ref{asm:iv} implies the counterfactual condition, i.e., $Y^{a=0,z=0} = Y^{a=0,z=1}$, where $Y^{a=0,z=0}$ and $Y^{a=0,z=1}$ are well-defined counterfactual outcomes such that by a consistency condition, $Y^{a=0} = Y^{a=0,z=0} (1 - Z) + Y^{a=0,z=1} Z$ almost surely. Importantly, note also that the second condition does not impose any restriction on $Y^{a=1}$; in particular, it allows for the presence of a direct effect of $Z$ on $Y^{a=1}$.
The corresponding directed acyclic graph (DAG) and Single World Intervention Graph (SWIG) are shown in Figure~\ref{fig:dag}.

\begin{figure}[ht]
    \centering
    \includegraphics[width=0.85\textwidth]{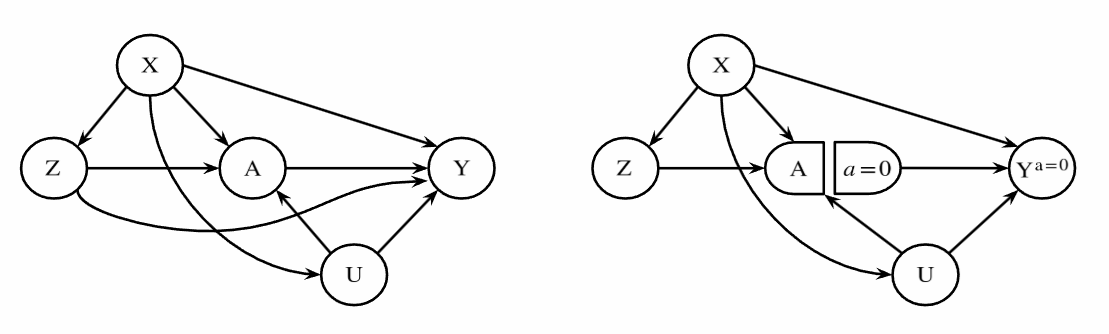}
    \caption{Causal diagram for the IV model.}
    \label{fig:dag}
\end{figure}

The task is to identify and estimate a functional of $P_{Y^{a=0} \mid A=1}$, the distribution of the counterfactual on the treated. Specifically, we consider any parameter $\beta^* \in \mathcal{B} \subset \R^d$ that can be represented as a solution of a moment equation
\begin{equation}\label{eqn:moment_equation}
    \EEst{M(Y^{a=0},\beta)}{A=1} = 0,
\end{equation}
where $M : \Y \times \mathcal{B} \rightarrow \R^d$ is a measureable moment function. Potential targets include:

\paragraph{Average treatment effect on the treated.}
Suppose we are interested in the average treatment effect on the treated (ATT), defined as  
\[\text{ATT} = \EEst{Y^{a=1} - Y^{a=0}}{A=1}.\]
Since learning $\EEst{Y^{a=1}}{A=1} = \EEst{Y}{A=1}$ is a trivial task, the main challenge is to estimate and infer $\EEst{Y^{a=0}}{A=1}$. Observe that this `mean of the counterfactual for the treated' can be represented as the solution to the moment equation~\eqref{eqn:moment_equation}, with $M(Y^{a=0}, \beta) = Y^{a=0} - \beta$.

\paragraph{Quantile treatment effect on the treated.}
Consider a setting where the target is the quantile treatment effect on the treated (QTT), given by
\begin{equation}\label{eqn:QTT}
    \text{QTT} = F_{Y^{a=1} \mid A=1}^{-1}(q) - F_{Y^{a=0} \mid A=1}^{-1}(q),
\end{equation}
where $q \in (0,1)$ is a predetermined target quantile. Here, to facilitate the exposition, we are assuming a setting where the counterfactual distributions are continuous, so that the inverse of the cdfs are well-defined, at least in the local neighborhood of the true quantile, although this can readily be relaxed. Again, the main task is to learn $F_{Y^{a=0} \mid A=1}^{-1}(q)$, which is the solution to the moment equation~\eqref{eqn:moment_equation}, with $M(Y^{a=0}, \beta) = \One{Y^{a=0} \geq \beta} - q$.

\paragraph{Cumulative distribution function of the counterfactual on the treated.}
Suppose the target is $\PPst{Y^{a=0} \leq y}{A=1}$ for a predetermined $y \in \R$. This corresponds to the moment equation~\eqref{eqn:moment_equation}, with $M(Y^{a=0},\beta) = \One{Y^{a=0} \leq y} - \beta$.\\

In this work, we examine how to learn such functionals of the counterfactual distribution within the context of the multiplicative IV setting, addressing the process from identification to efficient estimation and inference. 

\subsection{Main contributions}
Our contributions can be summarized as follows:

\paragraph{Identification and inference on general moment functions in the multiplicative IV model.}

We establish how the multiplicative IV model facilitates the identification of general counterfactual moment functions. We then introduce an efficient influence function-based estimator and demonstrate its multiple robustness property. We establish its asymptotic normality, which is then used as basis for robust inference about the mean of a general counterfactual moment function. 

\paragraph{Inference on the functional of the counterfactual distribution.}

We develop a framework for inference about a target functional of interest of the counterfactual distribution, defined implicitly as the solution to a counterfactual moment equation, which directly leverages the generality of the counterfactual moment functions-based approach we propose. Using the familiar statistical approach of inverting a test statistic, we produce confidence intervals for a large class of target functionals with a provable asymptotic coverage guarantee.

\paragraph{Empirical evaluation.}

We assess the quality of the estimator and the confidence interval through both simulations and an application to the Job Corps Study. The results indicate that the confidence interval achieves the coverage guarantee tightly, even in relatively small sample sizes.

\subsection{Related work}

The instrumental variables (IV) approach has a long-standing presence in econometrics, with its origins in the foundational works of \citet{wright1928tariff} and \citet{goldberger1972structural}, who first introduced the method within the framework of linear structural equations modeling. These pioneering contributions established a critical foundation for using instrumental variables to estimate causal effects. 
Subsequently, grounded in the potential outcome language of causality, \cite{angrist1995identification} and~\cite{angrist1996identification} established conditions for identification of so-called local average treatment effect; while~\cite{robins1994correcting} established sufficient conditions for identifying the average treatment effect on treatment on the treated, also see~\cite{hernan2006estimating}. More recently, in a separate strand of the IV literature, \cite{wang2018bounded} provide an alternative sufficient set of assumptions for identifying the population average treatment effect; these conditions were subsequently elaborated by~\cite{cui2021semiparametric}, and~\cite{qiu2021optimal}.  A recent overview of these modern IV developments and the corresponding semiparametric efficiency theory can be found in \citet{levis2024nonparametric}.

The statistical methods for making inferences about the target functionals in this paper are grounded in this semiparametric tradition, going back to  \citet{bickel1993efficient, tsiatis2006semiparametric}.  \cite{newey1990semiparametric} and~\cite{van1991differentiable}. 
Our specific implementation of semiparametric estimation for general nonparametric functionals follows closely from 
\citet{chernozhukov2018double} who propose a cross-fit scheme which combined with carefully constructed moment functions (mainly efficient influence functions) known to have a certain Neyman orthogonality property, accommodates the free-range use of modern machine methods for estimating nuisance functions, while minimizing their their downstream impact on inferences about the counterfactual target functionals. 

Several related papers study nonlinear treatment effects within the instrumental variable (IV) model. For instance, \citet{chernozhukov2017instrumental} introduced the Instrumental Variable Quantile Regression (IVQR) model to identify quantile treatment effects (QTEs). They achieve this by reinterpreting a key assumption—placing restrictions on the distribution of potential outcomes of the endogenous variable, assuming these distributions are identical—thus relaxing the rank invariance condition discussed in \citet{chernozhukov2005iv}. Additionally, \citet{abadie2002instrumental} examined a similar problem by identifying QTEs for the subpopulation of compliers. This approach accommodates essentially unrestricted heterogeneity in treatment effects, as it imposes no assumptions on the behavior of rank variables across potential treatment states. To ensure identification, they restrict attention to binary endogenous treatment settings and assume monotonicity in the relationship between the instrument and the treatment, in the tradition of local average treatment effects of Angrist and Imbens. Furthermore, \cite{chernozhukov2007instrumental} studied quantile treatment effects in nonseparable models by imposing a monotonicity assumption on the structural function with respect to the unobserved confounder $U$ to identify the functionals of interest.


For other nonlinear treatment effects; some recent works can be viewed as generalizations of \citet{wang2018bounded}. The first is \citet{michael2024instrumental} who considered identification in the IV framework with nonseparable structural  outcome model, by generalizing the assumption that the unmeasured confounder does not modify the additive effect of the instrument on treatment to accommodate time-varying treatment and instruments in a longitudinal setting. 
Additionally, \citet{michael2024instrumental} and \citet{mao2022identification} extend the ATE identification results of \citet{wang2018bounded} to the entire outcome distribution, allowing for robust identification of a broad class of nonlinear counterfactual functionals over the entire population, including both treated and untreated units (also see the technical report~\citet{tchetgen2018marginal}). In contrast, the current paper considers distinct identification conditions in a point instrument and treatment setting,  focusing primarily on identifying nonlinear counterfactual functionals for the treated group only---e.g., the median of the treatment-free counterfactual outcome of the treated. It is worth mentioning that our approach also differs from \cite{liu2020identification} who consider an IV framework in which any potential treatment effect on the treated, including nonlinear treatment effect such as quantile treatment effects, can be identified. Their identification strategy hinges upon a key assumption that the extended propensity score, $\PPst{A}{Y^{a=0}, Z, X}$, does not exhibit an interaction between $Y^{a=0}$ and $Z$ on the logit scale, a fundamentally different assumption to our MIV model.  

Several other papers have also investigated nonlinear treatment effects within the IV framework using structural (nested) mean models. For instance, \cite{vansteelandt2003causal} proposed generalized structural mean models to estimate causal effects when the outcome model follows a generalized linear model. Other notable contributions include \cite{robins2004estimation}, \cite{tan2010marginal}, and \cite{matsouaka2017instrumental}, among others. In addition, a line of research has focused on nonlinear effects in IV settings with survival-type outcomes; see~\cite{ying2023structural}, \cite{picciotto2012structural}, \cite{martinussen2019instrumental}, and \cite{wang2023instrumental} for further details.


Our work is closely related to and builds on the recent work of~\citet{liu2025multiplicativeinstrumentalvariablemodel}, who introduce the multiplicative IV model and discuss inference on the average treatment effect on the treated. We extend their identification results to general counterfactual functionals by developing a mapping from full-data moment equations whose solution implicitly defines the targeted counterfactual functional, to corresponding identifying observed-data moment equations through the construction of an appropriate set of IV-based weights. Crucially, our framework accommodates nonlinear estimand  that generally may not be available in closed form. To address this challenge for inference, we adopt a familiar statistical  test inversion strategy to construct confidence intervals for the functionals of interest.

An important  application that our proposed method covers is the quantile treatment effect, which has been previously studied by~\citet{kallus2024localized}, who propose a method for quantile causal effects (QTE) under unconfoundedness. They also discuss the IV setting under the monotonicity assumption, providing a method for QTE conditional on compliers. In contrast, our method provides inference on the \textit{quantile treatment effect on the treated} (QTT), under relatively mild assumptions.

\section{Main results}

We present the main result in the following structure. In Section~\ref{sec:moment_id}, we introduce the multiplicative IV model and discuss the identification of general moment equations. In Section~\ref{sec:moment_est}, we derive an efficient influence function-based estimator for the moment equation and discuss inference with asymptotic validity. Section~\ref{sec:functional} then discusses the inverse-inference procedure, which leads to inference on the target functional.

\subsection{Multiplicative IV model}\label{sec:moment_id}
For identification, we additionally assume the following multiplicative IV model.

\begin{assumption}[Multiplicative IV assumption]\label{asm:multi_iv}
    There exist functions $g_1 : \{0,1\} \times \X \rightarrow [0,1]$ and $g_2 : \mathcal{U} \times \X \rightarrow [0,1]$ such that the following holds.
    \[\PPst{A=1}{Z,U,X} = g_1(Z,X) \cdot g_2(U,X),\]
    with the only restriction $g_1(0,X)=1$ almost surely, to ensure uniqueness.
\end{assumption}
Assumption~\ref{asm:multi_iv} states that the instrumental variable $Z$ and the unmeasured confounder $U$ affect the treatment in a multiplicative manner. 

As noted earlier, different identification assumptions have been considered in the prior literature on instrumental variable methods, and Table~\ref{tab:IV_models} provides an overview and comparison with our setting---extending the one in~\citet{liu2025multiplicativeinstrumentalvariablemodel}. In Table~\ref{tab:IV_models}, we provide a generalized latent index model (GLIM) representation for each IV model, which is generally written in the form $A^z = \One{h(z,U) \geq \eps_z}$, $z=0,1$ (Here, $A^0$ and $A^1$ refer to the potential treatment values $A^{z=0}$ and $A^{z=1}$, respectively, under hypothetical interventions on the instrument). \citet{angrist1995identification} study the monotonicity assumption $A^{z=1} \geq A^{z=0}$---i.e., there are no defiers with ($A^{z=1} = 0$ and $A^{z=0} = 1$)---and establish that the average treatment effect on the compliers, $\EEst{Y^{a=1}-Y^{a=0}}{A^{z=1} = 1, A^{z=0} = 0}$, can be identified under this condition by a standard Wald ratio estimand; their result also implies identification of $\EEst{Y^{a}}{A^{z=1} = 1, A^{z=0} = 0}$ and in fact of any well-defined functional of the counterfactual distribution among the compliers, see for instance~\citet{abadie2003semiparametric}. A result due to~\citet{vytlacil2002independence} further shows that the monotonicity condition can equivalently be expressed as a specific GLIM,  $A^z = \One{g(z) + U \geq 1}$ with $g(1) \geq g(0)$ and $U \sim \textnormal{Unif}([0,1])$. \citep{wang2018bounded} consider a condition that includes the model $A^z = \One{g(z) + U \geq \eps_z}$ (with $\eps_z \sim \textnormal{Unif}([0,1])$) as a special case. They show that the average treatment effect---more generally, the conditional average treatment effect $\EEst{Y^{a=1} - Y^{a=0}}{X=x}$---can be identified under this condition. \citet{cui2021semiparametric} further show that the condition identifies any smooth nonlinear functional of the distribution of $Y^{a}$ in the population. \citet{tchetgen2024nudge} study inference on the nudge average treatment effect (NATE), defined as $\EEst{Y^{a=0}}{N=1}$, where $N = \One{A^{z=1} \neq A^{z=0}}$. For the identification of NATE, they consider a logistic model specified as
$A^z = \One{g(z) + U \geq \epsilon_z}$,
where $\epsilon_0, \epsilon_1 \iidsim \textnormal{Logistic}(0,1)$.

The multiplicative IV model, instead, includes the model $A^z = \One{g(z) \cdot U \geq \eps_z}$ (with $\eps_z \sim \textnormal{Unif}([0,1])$), and \cite{liu2025multiplicativeinstrumentalvariablemodel} shows that the average treatment effect on the treated, $\EEst{Y^{a=0}}{A=1}$, can be identified under this model, while the current work gives the corresponding identification resulf for any smooth functional of the treatment-free counterfactual in the treated. These different assumptions do not have a strict hierarchical relationship but instead offer distinct advantages depending on the context that may be helpful to communicate difference in identification conditions using a common framing. 

\begin{table}[ht]
\renewcommand{\arraystretch}{1.5}
\centering
\begin{tabular}{|c|c|c|}
\hline
 Identification assumption & Generalized latent index model &  Target of inference \\
\hline
\hline
 Monotonicity & \makecell{{}\\[-2mm]$A^z = \One{g(z) + U \geq \epsilon_z}$\\ $g(0), g(1) \in [0,1]$, $g(1) \geq g(0)$,\\
 $U \sim \textnormal{Unif}([0,1])$, $\eps_0 = \eps_1 = 1$. \vspace{2mm}} & \makecell{$\EEst{Y^{a=0}}{A^{z=1} = 1, A^{z=0} = 0}$. \\ \citep{angrist1995identification, vytlacil2002independence}}  \\
\hline
 Additive IV model & \makecell{{}\\[-2mm]$A^z = \One{g(z) + U \geq \epsilon_z}$, \\ For $z=0,1$, $\eps_z \sim \textnormal{Unif}([0,1])$, and \\ $g(z) + U \in [0,1]$ almost surely. \vspace{2mm}} & \makecell{$\EE{Y^{a=0}}$. \\ \citep{wang2018bounded}} \\
 \hline
 Logistic IV model & \makecell{{}\\[-2mm]$A^z = \One{g(z) + U \geq \epsilon_z}$, \\ $\eps_0, \eps_1 \iidsim \textnormal{Logistic}(0,1)$. \vspace{2mm}} & \makecell{$\EEst{Y^{a=0}}{N=1}$, \\ where $N = \One{A^{z=1} \neq A^{z=0}}$.\\ \citep{tchetgen2024nudge}} \\
\hline
 Multiplicative IV model & \makecell{{}\\[-2mm]$A^z = \One{g(z) \cdot U \geq \epsilon_z}$,\\ For $z=0,1$, $\eps_z \sim \textnormal{Unif}([0,1])$, and \\ $g(z) \cdot U \in [0,1]$ almost surely. \vspace{2mm}} & \makecell{$\EEst{Y^{a=0}}{A=1}$~\citep{liu2025multiplicativeinstrumentalvariablemodel}, \\[1mm] $\gamma(P_{Y^{a=0} \mid A=1})$, where $\gamma$ is a functional.
 \\ This manuscript} \\
\hline
\end{tabular}
\caption{Summary of various IV models, their corresponding generalized latent index models (GLIM), and possible targets of inference. For the monotonicity assumption and the Logistic IV model, the corresponding GLIM provides an equivalent representation, whereas for the Additive and Multiplicative IV models, the GLIM serves as an example. Note that the measured confounder $X$ is omitted in this summary---when it exists, we can regard it as being implicitly conditioned.}
\label{tab:IV_models}
\end{table}

We now show that the multiplicative IV model, in fact, enables identification for much broader targets. First, we derive the following result for identifying the moment equation~\eqref{eqn:moment_equation}.

\begin{theorem}\label{thm:identification}
Suppose Assumptions~\ref{asm:iv} and~\ref{asm:multi_iv} hold. Then
\[\EEst{M(Y^{a=0},\beta)}{A=1} = -\EEst{\frac{\delta^{M,A}(\beta,X)}{\delta^A(X)}}{A=1} = -\EE{\frac{\delta^{M,A}(\beta,X)}{\delta^A(X)} \cdot \frac{\PPst{A=1}{X}}{\PP{A=1}}},\]
where
\begin{equation}\label{eqn:nui_delta}
\begin{split}
    \delta^{M,A}(\beta,X) &= \EEst{M(Y,\beta) (1-A)}{Z=1,X} - \EEst{M(Y,\beta) (1-A)}{Z=0,X},\\
    \delta^A(X) &= \PPst{A=1}{Z=1,X} - \PPst{A=1}{Z=0,X}.
\end{split}
\end{equation}
    
\end{theorem}

By Theorem~\ref{thm:identification}, the target functional $\beta^*$ can be represented as the solution of
\begin{equation}\label{eqn:hbp}
    h(\beta,P) := \Ep{P}{\frac{\delta^{M,A}(\beta,X)}{\delta^A(X)} \cdot \PPst{A=1}{X}} = 0,
\end{equation}
under the multiplicative IV model. 

Theorem~\ref{thm:identification} illustrates the effectiveness of the multiplicative IV model. Specifically, it states that having a multiplicative IV allows for the identification of any function of the counterfactual outcome, as long as it is measurable. Consequently, this generality facilitates `inverse inference'---inference on the solution to the moment equation, which we will discuss in further detail in Section~\ref{sec:functional}.

\subsection{Efficient estimation and inference on the moment equation}\label{sec:moment_est}

Now, we construct an efficient estimator for the value of $h(\beta,P)$, as defined in~\eqref{eqn:hbp} under a nonparametric model for the observed data distribution. For the discussion in this section, we assume that $\beta$ is a pre-specified fixed parameter.

We first derive the efficient influence function of $h(\beta,P)$. Let us define the nuisance functions $\pi_z : \X \rightarrow [0,1], \mu_z : \R^d \times \X \rightarrow \R^d, \lambda_z : \X \rightarrow [0,1]$ as
\[\pi_z(X) = \PPst{Z=z}{X}, \mu_z(\beta,X) = \EEst{M( Y,\beta) (1-A)}{Z=z,X} \text{ and } \lambda_z(X) = \PPst{A=1}{Z=z,X},\]
for $z=0,1$, and $\rho(X) = \PPst{A=1}{X}$. Theorem~\ref{thm:eif} provides the efficient influence function for the moment equation, along with the formula for the second-order remainder term. For conciseness, we omit the inputs $X$ and $\beta$ for the nuisance functions in the expression of the remainder term.

\begin{theorem}\label{thm:eif}
    The efficient influence function of $h(\beta,P)$ in~\eqref{eqn:hbp} is given by
    \begin{multline}\label{eqn:eif_moment}
    \dot{h}(\beta,O,P) = (\delta(\beta,X)\cdot A - h(\beta,P)) + \frac{\rho(X)}{\delta^A(X)}\cdot\frac{2Z-1}{\pi_Z(X)}\\
    \cdot\Big(M( Y,\beta) (1-A) - \mu_Z(\beta,X) - (A - \lambda_Z(X))\cdot\delta(\beta,X)\Big),
\end{multline}
where $\delta(\beta,X) = \delta^{M,A}(\beta,X)/\delta^A(X)$. Moreover, for any distribution $\bar{P}$, the second-order remainder term
\begin{align*}
    R(\bar{P},P) &= h(\beta,\bar{P}) - h(\beta,P) + \Ep{P}{\dot{h}(\beta,O,\bar{P})}
\end{align*}
is equal to
\begin{multline*}
    \mathbb{E}_{P}\bigg[(\bar{\delta} - \delta)\left(\frac{\bar{\rho}}{\bar{\delta}^A}(\bar{\delta}^A - \delta^A) - (\bar{\rho}-\rho)\right)\\
    + \frac{\bar{\rho}}{\bar{\delta}^A} (\bar{\pi}_1-\pi_1)\left(\frac{(\bar{\mu}_1-\mu_1) - \bar{\delta}(\bar{\lambda}_1 - \lambda_1)}{\bar{\pi}_1} + \frac{(\bar{\mu}_0-\mu_0) - \bar{\delta}(\bar{\lambda}_0 - \lambda_0)}{1-\bar{\pi}_1}\right)\bigg],
\end{multline*}
    where $\bar{\mu}_z, \bar{\lambda}_z, \bar{\pi}_z$ and  $\bar{\rho}$ denote the nuisance functions corresponding to $\bar{P}$, and we write $\bar{\delta}^A = \bar{\lambda}_1 - \bar{\lambda}_0$ and $\bar{\delta} = \frac{\bar{\mu}_1 - \bar{\mu}_0}{\bar{\lambda}_1 - \bar{\lambda}_0}$.
\end{theorem}

From this result, we can derive the following multiple robustness structure.

\begin{corollary}[Multiple robustness structure]\label{cor:mult_robust}
The efficient influence function of $h(\beta,P)$ is unbiased if one of the following holds.
\begin{enumerate}\label{eqn:mult_robust}
    \item $\delta$ and $\pi_1$ are correctly specified.
    \item $\delta$, $\lambda_0$, and $\mu_0$ are correctly specified.
    \item $\delta^A$, $\pi_1$, and $\rho$ are correctly specified.
\end{enumerate}
\end{corollary}

We may proceed to construct a multiply robust semiparametric efficient estimator of $h(\beta,P)$ by solving an empirical version of the efficient influence function, which yields 
\begin{equation}\label{eqn:est_h}
\begin{split}
    \hat{h}(\beta) &= h(\beta,\hat{P}) + \Ep{n}{\dot{h}(\beta,O,\hat{P})}\\
    &= \Ep{n}{\hat{\delta}(\beta,X) A + \frac{\hat{\rho}(X)}{\hat{\delta}^A(X)}\cdot\frac{2Z-1}{\hat{\pi}_Z(X)}\cdot\left(M( Y,\beta) (1-A)-\hat{\mu}_Z(\beta,X) - \hat{\delta}(\beta,X)(A-\hat{\lambda}_Z(X))\right)},
\end{split}
\end{equation}
where $\Ep{n}{f(O)} = \frac{1}{n}\sum_{i=1}^n f(O_i)$ denotes the expectation with respect to the empirical distribution, and $\hat{P}$ represents the distribution determined by the nuisance estimators $\hat{\eta} = (\hat{\rho}, \hat{\pi}_1, \hat{\mu}_0, \hat{\mu}_1, \hat{\lambda}_0, \hat{\lambda}_1)$, which are constructed independently of the data. For example, we can split the data into two sets, using one to construct $\hat{P}$ and the other for estimation and inference. Later in this section, we discuss the application of cross-fitting, which allows the use of multiple folds with switching roles. As a remark, we will occasionally abuse the notation and write $\dot{h}(\beta, O, \hat{\eta})$ instead of $\dot{h}(\beta, O, \hat{P})$ when $\hat{P}$ represents the distribution corresponding to the nuisance functions $\hat{\eta}$.

We also show that under certain regularity conditions, the estimator $\hat{h}(\beta)$ has asymptotic normality:

\begin{corollary}\label{cor:asymp}
    If $\|h(\beta,\hat{P}) - h(\beta,P)\| = o_P(1)$ and $R(\hat{P},P) = o_P(n^{-1/2})$, then
    \[\sqrt{n}(\hat{h}(\beta) - h(\beta,P)) \xrightarrow{D} N(0,\textnormal{Var}_P(\dot{h}(\beta,O,P))). \]
\end{corollary}

Based on these observations, in the setting where $\beta \in \R$ and $h(\beta,P) \in \R$ are scalars, we can construct a confidence interval for $h(\beta,P)$ with asymptotic coverage as follows:
\begin{equation}\label{eqn:CI_moment}
    \chh(\beta) = \left[\hat{h}(\beta) - z_\alpha \cdot \frac{1}{\sqrt{n}}\cdot\sqrt{\Ep{n}{\dot{h}(\beta,O,\hat{P})^2}}, \hat{h}(\beta) + z_\alpha \cdot \frac{1}{\sqrt{n}}\cdot\sqrt{\Ep{n}{\dot{h}(\beta,O,\hat{P})^2}}
    \right],
\end{equation}
where $\alpha \in (0,1)$ is a predefined level.

 More generally, for multivariate $\beta$ and $h(\beta,P)$ in $\R^d$, one can construct the confidence interval for $w^\top h(\beta,P)$---where $w \in \R^d$ is a fixed vector---as follows.
\begin{equation}\label{eqn:CI_moment_mult}
    \chh^w(\beta) = \left[w^\top \hat{h}(\beta) - z_\alpha \cdot \frac{1}{\sqrt{n}}\cdot\sqrt{w^\top \Ep{n}{\dot{h}(\beta,O,\hat{P})^2} w}, w^\top\hat{h}(\beta) + z_\alpha \cdot \frac{1}{\sqrt{n}}\cdot\sqrt{w^\top \Ep{n}{\dot{h}(\beta,O,\hat{P})^2}w}
    \right],
\end{equation}

The above confidence intervals are asymptotically valid---in the sense that the coverage rate converges to the target level $1-\alpha$---under certain regularity conditions. We provide theoretical details in the following section, which discusses general procedure with cross-fitting based on multiple folds.

\subsubsection{Procedure with cross-fitting}

Cross-fitting~\citep{chernozhukov2018double} is a widely used technique in machine learning and causal inference that allows efficient use of data while preserving the validity of inference. Specifically, cross-fitting begins by constructing multiple data folds. The process then involves using one fold for inference, where the nuisance parameters involved are trained on all remaining folds. Finally, the inference results from different folds are aggregated by simple averaging. Cross-fitting ensures the validity of inference by separating the data used for training from that used for inference in each repetition, while improving the quality of estimation by effectively integrating all data points to construct the estimators.

Here, we present the procedure for performing inference on the functional using cross-fitting within the context of our multiplicative IV setting. For simplicity, we fix $\beta \in \mathcal{B}$ and write the target as $\theta^* = h(\beta,P)$. From now on, we denote the true nuisance functions as $\eta^* = (\rho^*,\pi_1^*, \mu_0^*, \mu_1^*, \lambda_0^*, \lambda_1^*)$ to avoid confusion.

We first construct multiple folds of the data $(O_i)_{i \in [n]}$ by constructing a partition $\{I_1,\cdots,I_K\}$ of $[n]$, i.e., $I_1 \cup \cdots \cup I_K = [n]$ and $I_k$s are disjoint, where we set each partition size to be $|I_k| = n/K$. Then for each $k=1,2,\cdots,K$, one repeat the following steps:
\begin{enumerate}
    \item Using the data $(O_i)_{i \in I_k^c}$, construct the nuisance estimators $\hat{\eta}_{-k} = (\hat{\rho}^{-k}, \hat{\pi}_1^{-k}, \hat{\mu}_0^{-k}, \hat{\mu}_1^{-k}, \hat{\lambda}_0^{-k}, \hat{\lambda}_1^{-k})$.

    \item Construct the estimator $\hat{\theta}_k$ using the data $(O_i)_{i \in I_k}$, based on~\eqref{eqn:est_h}.

    \item Construct the variance estimator $\hat{\sigma}_k^2 = \Ep{I_k}{\dot{h}(\beta,O,\hat{\eta}_{-k})\cdot \dot{h}(\beta,O,\hat{\eta}_{-k})^\top}$, where $\mathbb{E}_{I_k}$ denotes the average over the $k$-th fold.
\end{enumerate}
Then the final estimator for $\theta^* = h(\beta,P)$ is given as $\hat{\theta} = \frac{1}{K}\sum_{k=1}^K \hat{\theta}_k$,
and the confidence interval for $\theta^*$---in the case $\theta^* \in \R$---is constructed as
\begin{equation}\label{eqn:CI_crossfit}
    \chh(\beta) = \left[\hat{\theta} - z_\alpha \cdot \frac{\hat{\sigma}}{\sqrt{n}},\; \hat{\theta} + z_\alpha \cdot \frac{\hat{\sigma}}{\sqrt{n}}
    \right],\qquad\textnormal{ where } \hat{\sigma}^2 = \frac{1}{K}\sum_{k=1}^K \hat{\sigma}_k^2.
\end{equation}
Generally, for $\theta^* \in \R^d$, the confidence interval for $w^\top \theta^*$, where $w \in \R^d$ is an arbitrary vector, can be constructed as

\begin{equation}\label{eqn:CI_crossfit_w}
    \chh^w(\beta) = \left[w^\top\hat{\theta} - z_\alpha \cdot \sqrt{\frac{w^\top\hat{\Sigma} w}{n}},\; \hat{\theta} + z_\alpha \cdot \sqrt{\frac{w^\top\hat{\Sigma} w}{n}}
    \right],\qquad\textnormal{ where } \hat{\Sigma} = \frac{1}{K}\sum_{k=1}^K \hat{\sigma}_k^2.
\end{equation}

We now present the theoretical result for the asymptotic coverage of the aforementioned confidence interval. First, we state the regularity conditions:

\begin{assumption}\label{asm:regularity}
There exists a sequence of family of distributions $(\mathcal{P}_n)_{n \in \N}$ such that the following conditions hold:

\begin{enumerate}
    \item There exist constants $c_1, c_2, c_3, c_4, c_5, c_6 > 0$ and sequences $(\tau_n)_{n \in \mathbb{N}} \subset \mathbb{R}^+$, and $(\eps_n)_{n \in \mathbb{N}} \subset \mathbb{R}^+$ such that $\lim_{n \to \infty} \tau_n = \lim_{n \to \infty} \eps_n = 0$ and $\tau_n > 1/\sqrt{n}$ hold, and that for each $k \in [K]$, the nuisance estimator $\hat{\eta}_{-k}$ satisfies the following conditions for all $P \in \mathcal{P}_n$, with probability at least $1 - \eps_n$:
    \begin{enumerate}
        \item $\|\hat{\eta}_{-k}(X)\|_\infty \leq c_1 \quad \textnormal{almost surely under } P$.
        \item $\Ep{P}{\|\hat{\eta}_{-k} (X) - \eta^*(X)\|^2}^{1/2} \leq c_2 \cdot \tau_n$.
        \item The following product biases are bounded by $c_3\cdot \tau_n / \sqrt{n}$:
    \begin{alignat*}{3}
        &\|(\hat{\delta}^{-k} - \delta^*)({\delta^A}^{-k} - {\delta^A}^*)\|,\quad 
        &&\|(\hat{\delta}^{-k} - \delta^*)(\hat{\rho}^{-k} - \rho^*)\|,\quad 
        &&\|(\hat{\pi}_1^{-k} - \pi_1^*)(\hat{\mu}_1^{-k} - \mu_1^*)\|,\\
        &\|(\hat{\pi}_1^{-k} - \pi_1^*)(\hat{\mu}_0^{-k} - \mu_0^*)\|,\quad 
        &&\|(\hat{\pi}_1^{-k} - \pi_1^*)(\hat{\lambda}_1^{-k} - \lambda_1^*)\|,\quad 
        &&\|(\hat{\pi}_1^{-k} - \pi_1^*)(\hat{\lambda}_0^{-k} - \lambda_0^*)\|.
    \end{alignat*}
    \item $c_4 \leq \hat{\pi}_1^{-k}(X) \leq 1-c_5$ and $|\hat{\lambda}_1^{-k}(X) - \hat{\lambda}_0^{-k}(X)| \geq c_6$ hold almost surely under $P$.
    \end{enumerate}
    \item The moment function is uniformly bounded, i.e., there exists a constant $c$ such that $|M(Y,\beta)| \leq c$ almost surely under $P$, for all $P \in \mathcal{P}_n$.
    
    \item The variance of the efficient influence function is non-degenerate: there exists a constant $c' > 0$ such that
    \[\Ep{P}{\dot{h}(\beta,O,P) \dot{h}(\beta,O,P)^\top} \succeq c' I,\]
    for all $P \in \mathcal{P}_n$.
\end{enumerate}

\end{assumption}

Intuitively, the first condition in Assumption~\ref{asm:regularity} requires that the nuisance estimator is bounded, consistent in the $L^2$ norm, and that the second-order remainder term vanishes at a rate of at least $1/\sqrt{n}$---note that the product biases in condition 1(c) are exactly the terms that appear in the expression for the second-order bias $R(\hat{P}, P)$. For example, one can consider the case where all the nuisance estimators converge at a rate of $o(n^{-1/4})$. Additionally, even if some nuisance estimators converge at a slower rate than \( o(n^{-1/4}) \), the bias---being a product of estimation errors---can still be controlled, as the order faster-converging nuisance estimators may compensate for the slower ones, ensuring the overall bias remains of smaller order.

\begin{theorem}\label{thm:unif_conv}
Suppose Assumption~\ref{asm:regularity} hold. Then for any $w \in \R^d$, the confidence interval~\eqref{eqn:CI_crossfit_w} satisfies
\[\lim_{n \rightarrow \infty} \sup_{P \in \mathcal{P}_n} \left|\Pp{P}{w^\top h(\beta,P) \in \chh^w(\beta)} - (1-\alpha)\right| = 0.\]
\end{theorem}

A direct consequence of the uniform convergence result in Theorem~\ref{thm:unif_conv} is the following corollary. Basically, it states that if the nuisance estimator is consistent and has product biases of the scale $o(1/\sqrt{n})$, then the coverage rate of the confidence interval $\chh(\beta)$ converges to the target level.

\begin{corollary}\label{cor:moment_cov}
    Suppose that $h(\beta,P) \in \R$ and that Assumption~\ref{asm:regularity} holds for $\mathcal{P}_n \equiv \{P\}$. Then the confidence interval $\chh(\beta)$ from~\eqref{eqn:CI_crossfit} satisfies
    \[\lim_{n \rightarrow \infty}\PP{h(\beta,P) \in \chh(\beta)} = 1-\alpha.\]
\end{corollary}

\subsection{Inference on a functional of counterfactual distribution}\label{sec:functional}

We next discuss the estimation and inference for the parameter $\beta^*$, which is given as the solution of the moment equation $h(\beta,P) = 0$.
We first look into a few examples.

\paragraph{Analysis on the average treatment effect on the treated (ATT)}

Suppose our target is ATT, defined as $\EEst{Y^{a=1}-Y^{a=0}}{A=1}$. We consider the moment equation~\ref{eqn:moment_equation} with $M(Y^{a=0},\beta) = Y^{a=0} - \beta$, whose solution is given by $\beta^* = \EEst{Y^{a=0}}{A=1}$. In this case, we have $h(\beta,P) = h(0,P) - \beta\cdot \PP{A=1}$, indicating that $\beta^* = h(0,P) / \PP{A=1}$ and consequently
\[\text{ATT} = \EEst{Y}{A=1} - \frac{1}{\PP{A=1}} \cdot h(0,P).\]
Therefore, we have the following efficient estimator for ATT,
\[\widehat{\text{ATT}} = \Epst{n}{Y}{A=1} - \frac{1}{\Pp{n}{A=1}}\cdot\hat{h}(0),\]
and we can directly apply Corollary~\ref{cor:moment_cov} for the inference on ATT. This recovers the results discussed in \cite{liu2025multiplicativeinstrumentalvariablemodel}.

\paragraph{Analysis on the median effect on the treated}

Now suppose we are interested in $\beta^* = \text{median}(P_{Y^{a=0} \mid A=1})$, which is the solution of the moment equation with $M(Y^{a=0},\beta) = \One{Y^{a=0} \geq \beta} - \frac{1}{2}$. For this target, $\beta^*$ cannot be expressed as a simple formula with $h(\cdot,P)$, and thus the result in Section~\ref{sec:moment_est} for $h(\beta,P)$ does not directly lead to an inference on $\beta^*$.\\

Generally, when the target functional $\beta^*$ does not have a closed-form expression, as in the second example, one can construct the confidence interval (or more generally, a confidence set) for $\beta^*$ as follows:
\begin{equation}\label{eqn:CI_beta}
    \chb = \left\{\beta \in \mathcal{B} : 0 \in \chh(\beta)\right\} = \left\{\beta : |\hat{h}(\beta)| \leq z_\alpha \cdot \frac{\hat{\sigma}}{\sqrt{n}}\right\},
\end{equation}
where $\chh(\beta)$ follows the definition in~\eqref{eqn:CI_crossfit}. Generally, for multivariate target $\beta^* \in \R^d$, the confidence interval can be constructed as
\begin{equation}\label{eqn:CI_beta_d}
\chb = \left\{\beta : n \cdot \|\hat{\Sigma}^{-1/2} \cdot\hat{h}(\beta)\|^2 \leq Q_{1-\alpha}(\chi_d^2)\right\},
\end{equation}
where $Q_{1-\alpha}(\chi_d^2)$ denotes the $(1-\alpha)$-quantile of the $\chi^2$ distribution with degrees of freedom $d$.

In practice, the parameter space $\mathcal{B}$ is often not a finite set, and thus computing $\chh(\beta)$ for all $\beta$ is not possible. In such cases, one can instead perform the computation on the grid $\tilde{\mathcal{B}}$, a finite subset of $\mathcal{B}$. The overall process of the procedure for a one-dimensional target parameter is outlined in Algorithm~\ref{alg:functional}.

\begin{corollary}\label{cor:CI_beta}
    Under the assumptions of Corollary~\ref{cor:asymp}, the interval $\chb$ from~\eqref{eqn:CI_beta_d} satisfies
     \[\lim_{n \rightarrow \infty}\PP{\beta^* \in \chb} = 1-\alpha.\]
\end{corollary}

\begin{algorithm}
\caption{Inference on a functional in the Multiplicative IV model (one-dimensional target)}
\label{alg:functional}
\textbf{Input:} Dataset $(O_i)_{i \in [n]}$, target significance level $\alpha \in (0,1)$, moment function $M : \Y \times \mathcal{B} \rightarrow \R^d$, grid $\tilde{\mathcal{B}} \subset \mathcal{B}$.

\textbf{Step 1:} Split the dataset into $K$ folds, $I_1, I_2,\cdots,I_k$.

\textbf{Step 2:} For $k=1,2,\cdots,K$, construct the nuisance estimators $\hat{\eta}_{-k}$ using $(O_i)_{i \notin I_k}$.

\textbf{Step 3:} Repeat the following steps for $\beta \in \tilde{\mathcal{B}}$.

\textbf{Step 3-1:} For $k=1,2,\cdots,K$, repeat the following steps:
\begin{enumerate}
    \item Construct the estimator $\hat{\theta}_k$ for $h(\beta,P)$ using the data $(O_i)_{i \in I_k}$, following the expression in~\eqref{eqn:est_h}.

    \item Compute $\hat{\sigma}_k^2 = \Ep{I_k}{\dot{h}(\beta,O,\hat{\eta}_{-k})^2}$.
   
\end{enumerate}

\textbf{Step 3-2:} Compute the cross-fitting estimators $\hat{\theta}_\beta = \frac{1}{K}\sum_{k=1}^K \hat{\theta}_k$ and $\hat{\sigma}_\beta^2 = \frac{1}{K}\sum_{k=1}^K \hat{\sigma}_k^2$.

\textbf{Step 4:} Compute $\beta_L = \min\left\{\beta \in \tilde{\mathcal{B}} : |\hat{\theta}_\beta| \leq z_\alpha\cdot \frac{\hat{\sigma}_\beta}{\sqrt{n}}\right\}$ and $\beta_R = \min\left\{\beta \in \tilde{\mathcal{B}} : |\hat{\theta}_\beta| \leq z_\alpha\cdot \frac{\hat{\sigma}_\beta}{\sqrt{n}}\right\}$. 

\textbf{Output:} Confidence set $\chb = [\beta_L,\beta_R]$.

\end{algorithm}

\begin{remark}
The confidence interval $\chb$ in~\eqref{eqn:CI_beta} is not theoretically guaranteed to be an interval in the one-dimensional case. Consequently, the output of the grid-based approach described in Algorithm~\ref{alg:functional} may not provide an accurate approximation of $\chb$. We remark that this approach is suggested as a practical numerical method that generally performs well in practice.
\end{remark}

\subsection{Extension: inference on the quantile treatment effect on the treated}\label{sec:qtt}

In this section, we explore the problem of inference on the quantile treatment effect on the treated (QTT), $\xi^* = F_{Y^{a=1} \mid A=1}^{-1}(q) - F_{Y^{a=0} \mid A=1}^{-1}(q)$. Given a target quantile $q \in (0,1)$, the method~\ref{alg:functional} for inference on functionals can be applied to the counterfactual quantile $F_{Y^{a=0} \mid A=1}^{-1}(q)$, whereas inference on $F_{Y^{a=1} \mid A=1}^{-1}(q)$ can be carried out more directly using the samples $\{Y_i : A_i = 1\}$. But how can we obtain reliable inference for the difference between the two quantiles, rather than conducting separate inferences for each of them? A simple approach is to construct confidence intervals for each quantile separately and then apply the union bound to obtain a confidence interval for their difference, but this typically leads to a conservative interval. Here, we introduce an alternative approach inspired by the method of~\citet{berger1994p}.

We first construct a $(1-\zeta)$ confidence interval $\widehat{C}_{\gamma^*}$ for $\gamma^* := F_{Y^{a=1} \mid A=1}^{-1}(q)$ using the samples $\{Y_i : A_i = 1\}$, where $\zeta < \alpha$ is a predefined level. Then for different values of $\xi$, we compute
\[p(\xi) = \max_{\gamma \in \widehat{C}_{\gamma^*}} p(\xi;\gamma) + \zeta, \text{ where } p(\xi;\gamma) = 2(1-\Phi(\sqrt{n}\cdot|\hat{\theta}_{\xi+\gamma}|/\hat{\sigma}_{\xi+\gamma})),\]
where $\hat{\theta}_{\xi+\gamma}$ and $\hat{\sigma}_{\xi+\gamma}$ are computed according to the steps in Algorithm~\ref{alg:functional}, with $\beta = \xi+\gamma$. Intuitively, the term $p(\xi;\gamma)$ can be interpreted as a p-value for testing the hypothesis $H_{\xi+\gamma} : h(\xi+\gamma,P) = 0$, while $p(\xi)$ corresponds to a p-value for $H_{\xi+\gamma^*} : h(\xi+\gamma^*,P) = 0$---based on the result of~\citet{berger1994p}.
Then we construct the confidence set for the QTT $\xi^*$ as
\begin{equation}\label{eqn:CI_qtt}
    \widehat{C}_{\xi^*} = \left\{\xi : p(\xi) > \alpha\right\}.
\end{equation}

\begin{corollary}\label{cor:qtt}
    Under the assumptions of Corollary~\ref{cor:asymp}, The confidence set from~\eqref{eqn:CI_qtt} satisfies
    \[\lim_{n \rightarrow \infty}\PP{\xi^* \in \widehat{C}_{\xi^*}} \geq 1-\alpha.\]
\end{corollary}
The proof is deferred to the Appendix.

\section{Simulations}

In this section, we provide simulation results to illustrate the proposed method for the inference on the functional---specifically, the median of the counterfactual on the treated, $\beta^* = \text{median}(P_{Y^{a=0} \mid A=1})$. We generate the data as follows.

\begin{align*}
    &X = (X_1,X_2) \sim \textnormal{Unif}([0,1]^2),\qquad U \sim \textnormal{Unif}([0,1]),\\
    &Z \mid X \sim \textnormal{Bernoulli}\left(\frac{\exp(-1+\frac{1}{2}X_1 + X_2)}{1+\exp(-1+\frac{1}{2}X_1 + X_2)}\right),\\
    &A \mid Z,X,U \sim \textnormal{Bernoulli}\left(\exp\left(\frac{1+X_1+X_2+\frac{1}{2}X_1 X_2}{2}\cdot Z - X_1 - \frac{1}{2}X_2 - \frac{1}{3}U\right)\right),\\
    &Y^{a=1} \mid X,Z,U \sim N\left((X_1+\tfrac{1}{3}X_1^2+X_2+X_1X_2+Z)\cdot e^{U/3},1\right),\\
    &Y^{a=0} \sim N\left((X_1+X_2)\cdot e^{(U+\tfrac{1}{2})/5},1\right).
\end{align*}
We repeat the experiment with sample sizes $n = 400, 1000, 2000$. For each sample size $n$, we split the data into two parts, each of size $n/2$, and use the first part to fit the nuisance estimators via the super learner, which includes XGBoost and random forest as base algorithms. To approximate the true median $\beta^*$ for evaluation purposes, we generate a sample of size 5000 and compute the median of $Y^{a=0}$ values for which the treatment value is 1. We repeat the process of generating the data and checking whether the confidence intervals~\eqref{eqn:CI_beta} at levels $\alpha = 0.025, 0.05, 0.075, 0.1, \cdots, 0.2$ cover $\beta^*$, 500 times and compute the coverage rate---i.e., the proportion of trials in which $\beta^*$ is covered. The results are shown in Figure~\ref{fig:sim_1}, illustrating that the proposed confidence interval covers the target at the desired level tightly, even in the case of small sample sizes.

\begin{figure}[ht]
    \centering
    \includegraphics[width=0.9\textwidth]{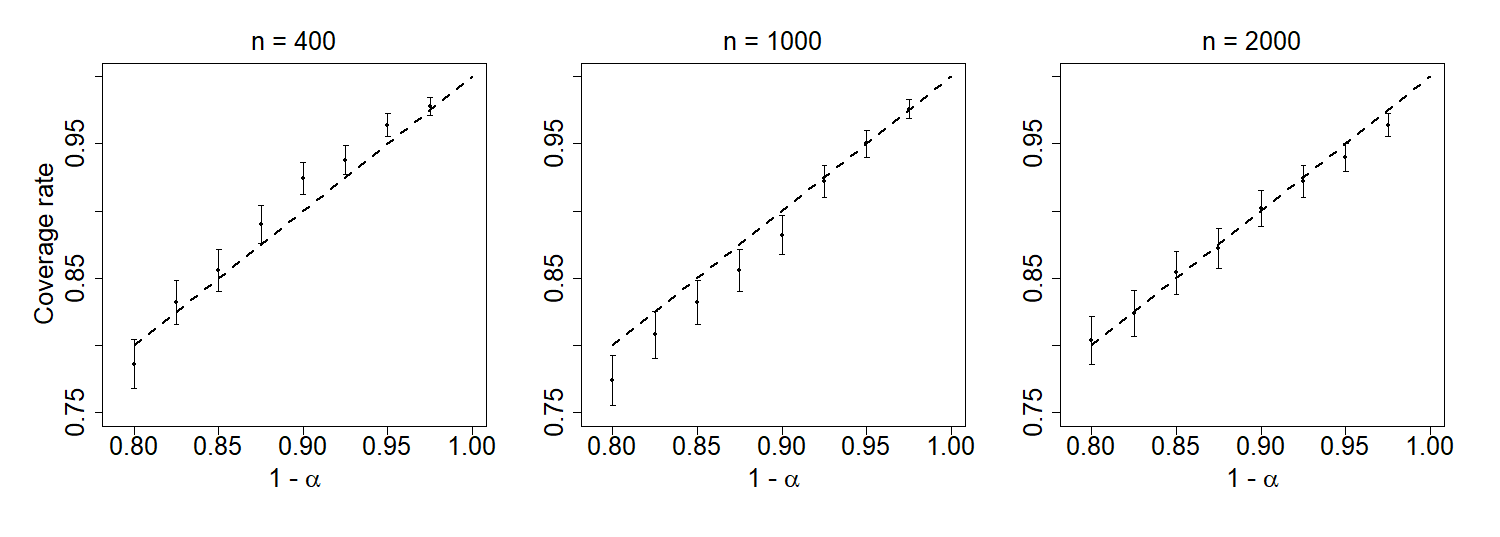}
    \caption{The coverage rate of the confidence interval~\eqref{eqn:CI_beta} for the median of the counterfactual on the treated with standard errors, from data splitting, at different levels. The dotted line corresponds to the $y=x$ line.}
    \label{fig:sim_1}
\end{figure}

Next, we repeat the simulation with cross-fitting. For each sample size $n$, we follow the same steps as before but include an additional step of switching the roles of the two splits. The results are shown in Figure~\ref{fig:sim_2}, illustrating that the proposed confidence interval again achieves the target coverage rate.

\begin{figure}[ht]
    \centering
    \includegraphics[width=0.9\textwidth]{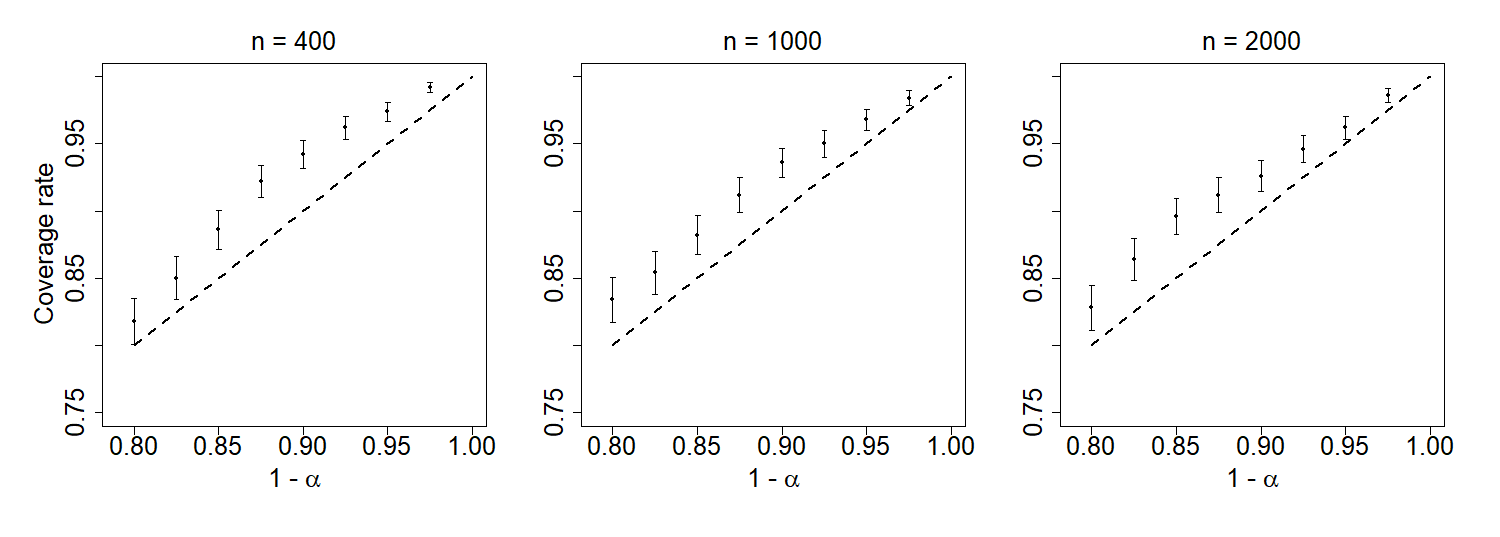}
    \caption{The coverage rate of the confidence interval~\eqref{eqn:CI_beta} for the median of the counterfactual on the treated with standard errors, from cross-fitting, at different levels. The dotted line corresponds to the $y=x$ line.}
    \label{fig:sim_2}
\end{figure}

\section{Application to Job Corps dataset}

We apply the proposed procedure to the Job Corps dataset~\citet{schochet2001national}. This dataset consists of 9,240 individuals who, after receiving an initial assignment, were given the option to participate in the job training program or not. We consider the initial assignment as the instrumental variable and the individuals' actual decision as the treatment, aiming to conduct inference on the median treatment effect on the treated, with the outcome being weekly earnings in the fourth year after assignment.
 We use 27 covariates, including demographic information as well as data on health, education, and employment.

We apply the QTT procedure described in Section~\ref{sec:qtt}, using the transformed outcome $\tilde{Y} = \log(1+Y)$. We use two-fold cross-fitting, where the training step applies the super learner, with details identical to the steps in the simulations. We apply the procedure with the grid from -0.5 to 0.5 with step size 0.01. Table~\ref{table:JC} reports the resulting confidence intervals at levels $\alpha = 0.05, 0.1, 0.2$, showing that, similar to the mean~\citep{liu2025multiplicativeinstrumentalvariablemodel}, the median of the weekly earnings is also positively affected by job training which is empirically found to have a beneficial effect.

\begin{table}[htbp]
\begin{center} 
\begin{tabular}{rlll}
\hline
Level & $\alpha=0.05$ & $\alpha=0.1$  & $\alpha=0.2$\\
\hline
Confidence interval & $[0.01,0.20]$ &  $[0.01,0.17]$ & $[0.03,0.13]$\\
\hline
\end{tabular} 
\end{center}
\caption{Results for the Job Corps dataset: confidence intervals for the median treatment effect on the treated from procedure~\ref{alg:functional}, at levels $\alpha=0.05,0.1$ and $0.2$.}
\label{table:JC}
\end{table}

\section{Discussion}

In this work, we discussed methods for inference on nonlinear counterfactual functionals under the multiplicative IV model. Our methods provide efficient and multiply robust estimators, along with asymptotically valid inference procedures. Experimental results for inference on the counterfactual median demonstrate that the proposed method performs well even with moderate sample sizes.

Given the broad applicability of the multiplicative IV model, many important questions remain open. For example, can the results be extended to settings where the instrumental variable or the treatment variable is non-binary, or even continuous? Developing methods with finite-sample validity is also an important potential direction, and we leave these questions to future work.


\bibliographystyle{plainnat}
\bibliography{bib}

\begin{thebibliography}{40}
\providecommand{\natexlab}[1]{#1}
\providecommand{\url}[1]{\texttt{#1}}
\expandafter\ifx\csname urlstyle\endcsname\relax
  \providecommand{\doi}[1]{doi: #1}\else
  \providecommand{\doi}{doi: \begingroup \urlstyle{rm}\Url}\fi

\bibitem[Abadie(2003)]{abadie2003semiparametric}
Alberto Abadie.
\newblock Semiparametric instrumental variable estimation of treatment response models.
\newblock \emph{Journal of econometrics}, 113\penalty0 (2):\penalty0 231--263, 2003.

\bibitem[Abadie et~al.(2002)Abadie, Angrist, and Imbens]{abadie2002instrumental}
Alberto Abadie, Joshua Angrist, and Guido Imbens.
\newblock Instrumental variables estimates of the effect of subsidized training on the quantiles of trainee earnings.
\newblock \emph{Econometrica}, 70\penalty0 (1):\penalty0 91--117, 2002.

\bibitem[Angrist and Imbens(1995)]{angrist1995identification}
Joshua~D Angrist and Guido~W Imbens.
\newblock Identification and estimation of local average treatment effects, 1995.

\bibitem[Angrist et~al.(1996)Angrist, Imbens, and Rubin]{angrist1996identification}
Joshua~D Angrist, Guido~W Imbens, and Donald~B Rubin.
\newblock Identification of causal effects using instrumental variables.
\newblock \emph{Journal of the American statistical Association}, 91\penalty0 (434):\penalty0 444--455, 1996.

\bibitem[Berger and Boos(1994)]{berger1994p}
Roger~L Berger and Dennis~D Boos.
\newblock P values maximized over a confidence set for the nuisance parameter.
\newblock \emph{Journal of the American Statistical Association}, 89\penalty0 (427):\penalty0 1012--1016, 1994.

\bibitem[Bickel et~al.(1993)Bickel, Klaassen, Bickel, Ritov, Klaassen, Wellner, and Ritov]{bickel1993efficient}
Peter~J Bickel, Chris~AJ Klaassen, Peter~J Bickel, Ya’acov Ritov, J~Klaassen, Jon~A Wellner, and YA'Acov Ritov.
\newblock \emph{Efficient and adaptive estimation for semiparametric models}, volume~4.
\newblock Springer, 1993.

\bibitem[Chernozhukov and Hansen(2005)]{chernozhukov2005iv}
Victor Chernozhukov and Christian Hansen.
\newblock An iv model of quantile treatment effects.
\newblock \emph{Econometrica}, 73\penalty0 (1):\penalty0 245--261, 2005.

\bibitem[Chernozhukov et~al.(2007)Chernozhukov, Imbens, and Newey]{chernozhukov2007instrumental}
Victor Chernozhukov, Guido~W Imbens, and Whitney~K Newey.
\newblock Instrumental variable estimation of nonseparable models.
\newblock \emph{Journal of Econometrics}, 139\penalty0 (1):\penalty0 4--14, 2007.

\bibitem[Chernozhukov et~al.(2017)Chernozhukov, Hansen, and W{\"u}thrich]{chernozhukov2017instrumental}
Victor Chernozhukov, Christian Hansen, and Kaspar W{\"u}thrich.
\newblock Instrumental variable quantile regression.
\newblock \emph{Handbook of quantile regression}, pages 119--143, 2017.

\bibitem[Chernozhukov et~al.(2018)Chernozhukov, Chetverikov, Demirer, Duflo, Hansen, Newey, and Robins]{chernozhukov2018double}
Victor Chernozhukov, Denis Chetverikov, Mert Demirer, Esther Duflo, Christian Hansen, Whitney Newey, and James Robins.
\newblock Double/debiased machine learning for treatment and structural parameters, 2018.

\bibitem[Cui and Tchetgen~Tchetgen(2021)]{cui2021semiparametric}
Yifan Cui and Eric Tchetgen~Tchetgen.
\newblock A semiparametric instrumental variable approach to optimal treatment regimes under endogeneity.
\newblock \emph{Journal of the American Statistical Association}, 116\penalty0 (533):\penalty0 162--173, 2021.

\bibitem[Goldberger(1972)]{goldberger1972structural}
Arthur~S Goldberger.
\newblock Structural equation methods in the social sciences.
\newblock \emph{Econometrica: Journal of the Econometric Society}, pages 979--1001, 1972.

\bibitem[Hern{\'a}n and Robins(2006)]{hernan2006estimating}
Miguel~A Hern{\'a}n and James~M Robins.
\newblock Estimating causal effects from epidemiological data.
\newblock \emph{Journal of Epidemiology \& Community Health}, 60\penalty0 (7):\penalty0 578--586, 2006.

\bibitem[Kallus et~al.(2024)Kallus, Mao, and Uehara]{kallus2024localized}
Nathan Kallus, Xiaojie Mao, and Masatoshi Uehara.
\newblock Localized debiased machine learning: Efficient inference on quantile treatment effects and beyond.
\newblock \emph{Journal of Machine Learning Research}, 25\penalty0 (16):\penalty0 1--59, 2024.

\bibitem[Kennedy(2016)]{kennedy2016semiparametric}
Edward~H Kennedy.
\newblock Semiparametric theory and empirical processes in causal inference.
\newblock \emph{Statistical causal inferences and their applications in public health research}, pages 141--167, 2016.

\bibitem[Kennedy et~al.(2020)Kennedy, Balakrishnan, and G’Sell]{kennedy2020sharp}
Edward~H Kennedy, Sivaraman Balakrishnan, and Max G’Sell.
\newblock Sharp instruments for classifying compliers and generalizing causal effects.
\newblock \emph{The Annals of Statistics}, 48\penalty0 (4):\penalty0 2008--2030, 2020.

\bibitem[Levis et~al.(2024)Levis, Kennedy, and Keele]{levis2024nonparametric}
Alexander~W Levis, Edward~H Kennedy, and Luke Keele.
\newblock Nonparametric identification and efficient estimation of causal effects with instrumental variables.
\newblock \emph{arXiv preprint arXiv:2402.09332}, 2024.

\bibitem[Liu et~al.(2025)Liu, Park, Lee, Zhang, Yu, Robins, and Tchetgen]{liu2025multiplicativeinstrumentalvariablemodel}
Jiewen Liu, Chan Park, Yonghoon Lee, Yunshu Zhang, Mengxin Yu, James~M. Robins, and Eric J.~Tchetgen Tchetgen.
\newblock The multiplicative instrumental variable model, 2025.
\newblock URL \url{https://arxiv.org/abs/2507.09302}.

\bibitem[Liu et~al.(2020)Liu, Miao, Sun, Robins, and Tchetgen]{liu2020identification}
Lan Liu, Wang Miao, Baoluo Sun, James Robins, and Eric~Tchetgen Tchetgen.
\newblock Identification and inference for marginal average treatment effect on the treated with an instrumental variable.
\newblock \emph{Statistica sinica}, 30\penalty0 (3):\penalty0 1517, 2020.

\bibitem[Mao(2022)]{mao2022identification}
Lu~Mao.
\newblock Identification of the outcome distribution and sensitivity analysis under weak confounder--instrument interaction.
\newblock \emph{Statistics \& probability letters}, 189:\penalty0 109590, 2022.

\bibitem[Martinussen et~al.(2019)Martinussen, N{\o}rbo~S{\o}rensen, and Vansteelandt]{martinussen2019instrumental}
Torben Martinussen, Ditte N{\o}rbo~S{\o}rensen, and Stijn Vansteelandt.
\newblock Instrumental variables estimation under a structural cox model.
\newblock \emph{Biostatistics}, 20\penalty0 (1):\penalty0 65--79, 2019.

\bibitem[Matsouaka and Tchetgen~Tchetgen(2017)]{matsouaka2017instrumental}
Roland~A Matsouaka and Eric~J Tchetgen~Tchetgen.
\newblock Instrumental variable estimation of causal odds ratios using structural nested mean models.
\newblock \emph{Biostatistics}, 18\penalty0 (3):\penalty0 465--476, 2017.

\bibitem[Michael et~al.(2024)Michael, Cui, Lorch, and Tchetgen~Tchetgen]{michael2024instrumental}
Haben Michael, Yifan Cui, Scott~A Lorch, and Eric~J Tchetgen~Tchetgen.
\newblock Instrumental variable estimation of marginal structural mean models for time-varying treatment.
\newblock \emph{Journal of the American Statistical Association}, 119\penalty0 (546):\penalty0 1240--1251, 2024.

\bibitem[Newey(1990)]{newey1990semiparametric}
Whitney~K Newey.
\newblock Semiparametric efficiency bounds.
\newblock \emph{Journal of applied econometrics}, 5\penalty0 (2):\penalty0 99--135, 1990.

\bibitem[Picciotto et~al.(2012)Picciotto, Hern{\'a}n, Page, Young, and Robins]{picciotto2012structural}
Sally Picciotto, Miguel~A Hern{\'a}n, John~H Page, Jessica~G Young, and James~M Robins.
\newblock Structural nested cumulative failure time models to estimate the effects of interventions.
\newblock \emph{Journal of the American Statistical Association}, 107\penalty0 (499):\penalty0 886--900, 2012.

\bibitem[Qiu et~al.(2021)Qiu, Carone, Sadikova, Petukhova, Kessler, and Luedtke]{qiu2021optimal}
Hongxiang Qiu, Marco Carone, Ekaterina Sadikova, Maria Petukhova, Ronald~C Kessler, and Alex Luedtke.
\newblock Optimal individualized decision rules using instrumental variable methods.
\newblock \emph{Journal of the American Statistical Association}, 116\penalty0 (533):\penalty0 174--191, 2021.

\bibitem[Robins and Rotnitzky(2004)]{robins2004estimation}
James Robins and Andrea Rotnitzky.
\newblock Estimation of treatment effects in randomised trials with non-compliance and a dichotomous outcome using structural mean models.
\newblock \emph{Biometrika}, 91\penalty0 (4):\penalty0 763--783, 2004.

\bibitem[Robins(1994)]{robins1994correcting}
James~M Robins.
\newblock Correcting for non-compliance in randomized trials using structural nested mean models.
\newblock \emph{Communications in Statistics-Theory and methods}, 23\penalty0 (8):\penalty0 2379--2412, 1994.

\bibitem[Schochet(2001)]{schochet2001national}
Peter~Z Schochet.
\newblock \emph{National Job Corps Study: The impacts of Job Corps on participants' employment and related outcomes}.
\newblock US Department of Labor, Employment and Training Administration, Office of~…, 2001.

\bibitem[Tan(2010)]{tan2010marginal}
Zhiqiang Tan.
\newblock Marginal and nested structural models using instrumental variables.
\newblock \emph{Journal of the American Statistical Association}, 105\penalty0 (489):\penalty0 157--169, 2010.

\bibitem[Tchetgen(2024)]{tchetgen2024nudge}
Eric J~Tchetgen Tchetgen.
\newblock The nudge average treatment effect.
\newblock \emph{arXiv preprint arXiv:2410.23590}, 2024.

\bibitem[Tchetgen et~al.(2018)Tchetgen, Michael, and Cui]{tchetgen2018marginal}
Eric J~Tchetgen Tchetgen, Haben Michael, and Yifan Cui.
\newblock Marginal structural models for time-varying endogenous treatments: A time-varying instrumental variable approach.
\newblock \emph{arXiv preprint arXiv:1809.05422}, 2018.

\bibitem[Tsiatis(2006)]{tsiatis2006semiparametric}
Anastasios~A Tsiatis.
\newblock \emph{Semiparametric theory and missing data}, volume~4.
\newblock Springer, 2006.

\bibitem[Van Der~Vaart(1991)]{van1991differentiable}
Aad Van Der~Vaart.
\newblock On differentiable functionals.
\newblock \emph{The Annals of Statistics}, pages 178--204, 1991.

\bibitem[Vansteelandt and Goetghebeur(2003)]{vansteelandt2003causal}
Stijn Vansteelandt and Els Goetghebeur.
\newblock Causal inference with generalized structural mean models.
\newblock \emph{Journal of the Royal Statistical Society Series B: Statistical Methodology}, 65\penalty0 (4):\penalty0 817--835, 2003.

\bibitem[Vytlacil(2002)]{vytlacil2002independence}
Edward Vytlacil.
\newblock Independence, monotonicity, and latent index models: An equivalence result.
\newblock \emph{Econometrica}, 70\penalty0 (1):\penalty0 331--341, 2002.

\bibitem[Wang and Tchetgen~Tchetgen(2018)]{wang2018bounded}
Linbo Wang and Eric Tchetgen~Tchetgen.
\newblock Bounded, efficient and multiply robust estimation of average treatment effects using instrumental variables.
\newblock \emph{Journal of the Royal Statistical Society Series B: Statistical Methodology}, 80\penalty0 (3):\penalty0 531--550, 2018.

\bibitem[Wang et~al.(2023)Wang, Tchetgen~Tchetgen, Martinussen, and Vansteelandt]{wang2023instrumental}
Linbo Wang, Eric Tchetgen~Tchetgen, Torben Martinussen, and Stijn Vansteelandt.
\newblock Instrumental variable estimation of the causal hazard ratio.
\newblock \emph{Biometrics}, 79\penalty0 (2):\penalty0 539--550, 2023.

\bibitem[Wright(1928)]{wright1928tariff}
Philip~Green Wright.
\newblock \emph{The tariff on animal and vegetable oils}.
\newblock Macmillan, 1928.

\bibitem[Ying and Tchetgen(2023)]{ying2023structural}
Andrew Ying and Eric J~Tchetgen Tchetgen.
\newblock Structural cumulative survival models for estimation of treatment effects accounting for treatment switching in randomized experiments.
\newblock \emph{Biometrics}, 79\penalty0 (3):\penalty0 1597--1609, 2023.

\end{thebibliography}

\newpage

\appendix

\section{Direct identification of quantile treatment effects}

In this section, we discuss an extension of the proposed method for inference on quantile treatment effects, where the corresponding moment function depends on unknown information. Suppose we are interested in the following quantile treatment effect on the treated:
\begin{equation}\label{eqn:qte}
    \beta^* = F_{Y^{a=0} \mid A=1}^{-1}(q) - F_{Y^{a=1} \mid A=1}^{-1}(q),
\end{equation}
where $q \in (0,1)$ represents a predetermined quantile of interest. A straightforward approach is to conduct inference on $F_{Y^{a=0} \mid A=1}^{-1}(q)$ based on the discussion in Section~\ref{sec:functional}, separately carry out inference on $F_{Y^{a=1} \mid A=1}^{-1}(q)$ using the samples $\{Y_i : A_i = 1\}$, and then combine the results.

Here, we introduce an approach for direct inference on $\beta^*$ using its moment function. Let us define $\gamma^* = F_{Y^{a=1} \mid A=1}^{-1}(q)$, and let
\begin{equation}\label{eqn:moment_qte}
    M(Y,X,Z,A,\beta) = \left(\One{Y \leq \gamma^*+\beta(1-A)} - q\right)\cdot\left(\frac{2Z-1}{\pi_Z(X)}\cdot\frac{\rho(X)}{\delta^A(X)}\right)^{1-A}.
\end{equation}
The nuisance functions $\pi_Z$, $\rho$, and $\delta^A$ are defined as before, while $\gamma^*$ is viewed as a nuisance parameter in this context. We now prove the following identification result.

\begin{proposition}\label{prop:qte}
    The quantile treatment effect $\beta^*$ from~\eqref{eqn:qte} satisfies
    \[\EE{M(Y,X,Z,A,\beta^*)} = 0,\]
    where $M$ is defined as~\eqref{eqn:moment_qte}.
\end{proposition}

\section{Additional proofs}\label{sec:proofs}

\subsection{Proof of Theorem~\ref{thm:identification}}
The second equality holds since
\begin{multline*}
\EE{\frac{\delta^{M,A}(\beta,X)}{\delta^A(X)} \cdot \frac{\PPst{A=1}{X}}{\PP{A=1}}} = \EE{\frac{\delta^{M,A}(\beta,X)}{\delta^A(X)} \cdot \frac{\One{A=1}}{\PP{A=1}}}\\
= \sum_{a \in \{0,1\}} \EEst{\frac{\delta^{M,A}(\beta,X)}{\delta^A(X)} \cdot \frac{\One{A=1}}{\PP{A=1}}}{A=a} \cdot \PP{A=a} = \EEst{\frac{\delta^{M,A}(\beta,X)}{\delta^A(X)} \cdot \One{A=1}}{A=1}\\
= \EEst{\frac{\delta^{M,A}(\beta,X)}{\delta^A(X)}}{A=1}.
\end{multline*}
Now we prove the first equality. First observe that $M(Y ,\beta)(1-A) = M(Y^{a=0} ,\beta)(1-A)$ holds almost surely by the consistency assumption. Therefore,
\begin{align*}
    &\EEst{M(Y,\beta) (1-A)}{Z=1,X}\\
    &= \EEst{M(Y^{a=0},\beta) (1-A)}{Z=1,X}\\
    &= \EEst{\EEst{M(Y^{a=0},\beta) \One{A=0}}{Z=1,A,U,X}}{Z=1,X}\\
    &= \EEst{\EEst{M(Y^{a=0},\beta)}{Z=1,A=0,U,X} \PPst{A=0}{Z=1,U,X}}{Z=1,X}\\
    &= \EEst{\EEst{M(Y^{a=0},\beta)}{U,X} \PPst{A=0}{Z=1,U,X}}{Z=1,X} \qquad \textnormal{by weak ignorability in Assumption~\ref{asm:iv}}\\
    &= \EEst{\EEst{M(Y^{a=0},\beta)}{U,X} \PPst{A=0}{Z=1, U,X}}{X} \qquad \textnormal{ since } U \indep Z \mid X.
\end{align*}
It follows that
\begin{align*}
    &\delta^{M,A}(\beta,X)\\
    &= \EEst{\EEst{M(Y^{a=0},\beta)}{U,X} \cdot ( \PPst{A=0}{Z=1, U,X} - \PPst{A=0}{Z=0, U,X})}{X}\\
    &= \EEst{\EEst{M(Y^{a=0},\beta)}{U,X} \cdot ( \PPst{A=1}{Z=0, U,X} - \PPst{A=1}{Z=1, U,X})}{X}\\
    &= \EEst{\EEst{M(Y^{a=0},\beta)}{U,X} \cdot \left(1 -  \frac{\PPst{A=1}{Z=1, U,X}}{\PPst{A=1}{Z=0, U,X}}\right)\cdot \PPst{A=1}{Z=0, U,X}}{X}\\
    &= \EEst{\EEst{M(Y^{a=0},\beta)}{U,X} \cdot \left(1 -  \frac{\PPst{A=1}{Z=1,X}}{\PPst{A=1}{Z=0,X}}\right)\cdot \PPst{A=1}{Z=0, U,X}}{X} \quad \text{ by Assumption~\ref{asm:multi_iv}}\\
    &=\EEst{\EEst{M(Y^{a=0},\beta)}{U,X} \cdot \PPst{A=1}{Z=0, U,X}}{X} \cdot \left(1 -  \frac{\PPst{A=1}{Z=1,X}}{\PPst{A=1}{Z=0,X}}\right).\\
\end{align*}
Note that in the fourth equality, we apply the following result, which follows from Assumption~\ref{asm:multi_iv}.
\begin{multline*}
    \frac{\PPst{A=1}{Z=1,X}}{\PPst{A=1}{Z=0,X}} = \frac{\EEst{\PPst{A=1}{Z=1,U,X}}{Z=1,X}}{\EEst{\PPst{A=1}{Z=0,U,X}}{Z=0,X}} = \frac{\EEst{g_1(1,X)g_2(U,X)}{Z=1,X}}{\EEst{g_1(0,X)g_2(U,X)}{Z=0,X}}\\
    = \frac{g_1(1,X) \cdot \EEst{g_2(U,X)}{Z=1,X}}{g_1(0,X) \cdot \EEst{g_2(U,X)}{Z=0,X}} = \frac{g_1(1,X)}{g_1(0,X)} = \frac{\PPst{A=1}{Z=1, U,X}}{\PPst{A=1}{Z=0, U,X}},
\end{multline*}
where the second-to-last step holds since $U \indep Z \mid X$.

Next, we prove the following conditional independence:
\begin{equation}\label{eqn:cond_indep}
    U \indep Z \mid A=1, X.
\end{equation}
Let $S$ be any measureable subset of $\mathcal{U}.$ For $z \in \{0,1\}$, we have
\begin{align*}
    &\PPst{U \in S}{A=1,Z=z,X} = \frac{\PPst{U \in S, A=1, Z=z}{X}}{\PPst{A=1,Z=z}{X}} = \frac{\EEst{\One{U \in S}\cdot\PPst{A=1,Z=z}{U,X}}{X}}{\EEst{\PPst{A=1,Z=z}{U,X}}{X}}\\
    &= \frac{\EEst{\One{U \in S}\cdot\PPst{A=1}{Z=z,U,X}\cdot \PPst{Z=z}{U,X}}{X}}{\EEst{\PPst{A=1}{Z=z,U,X}\cdot \PPst{Z=z}{U,X}}{X}}\\
    &= \frac{\EEst{\One{U \in S}\cdot\PPst{A=1}{Z=z,U,X}\cdot \PPst{Z=z}{X}}{X}}{\EEst{\PPst{A=1}{Z=z,U,X}\cdot \PPst{Z=z}{X}}{X}}\quad\textnormal{ since } Z \indep U \mid X\\
    &=\frac{\EEst{\One{U \in S}\cdot g_1(z,X)g_2(U,X) \cdot \PPst{Z=z}{X}}{X}}{\EEst{g_1(z,X)g_2(U,X)\cdot \PPst{Z=z}{X}}{X}}\quad\textnormal{ by Assumption~\ref{asm:multi_iv}}\\
    &= \frac{\EEst{\One{U \in S}\cdot g_2(U,X)}{X} \cdot g_1(z,X) \cdot \PPst{Z=z}{X}}{\EEst{g_2(U,X)}{X}\cdot g_1(z,X) \cdot \PPst{Z=z}{X}} = \frac{\EEst{\One{U \in S}\cdot g_2(U,X)}{X}}{\EEst{g_2(U,X)}{X}}.
\end{align*}
Since the resulting term does not depend on $z$, this implies the conditional independence~\eqref{eqn:cond_indep}. It follows that
\begin{align*}
    &\EEst{\EEst{M(Y^{a=0},\beta)}{U,X} \cdot \PPst{A=1}{Z=0, U,X}}{X}\\
    &=\EEst{\EEst{M(Y^{a=0},\beta)}{A=1,U,X} \cdot \PPst{A=1}{Z=0, U,X}}{X} \qquad \textnormal{ since } Y^{a=0} \indep A \mid U,X\\
    &=\EEst{\EEst{M(Y^{a=0},\beta)}{A=1,U,X} \cdot \PPst{Z=0}{A=1, U,X}\cdot\frac{\PPst{A=1}{U,X}}{\PPst{Z=0}{U,X}}}{X}\\
    &=\EEst{\EEst{M(Y^{a=0},\beta)}{A=1,U,X} \cdot \PPst{Z=0}{A=1, X}\cdot\frac{\PPst{A=1}{U,X}}{\PPst{Z=0}{X}}}{X}\quad\textnormal{ by~\eqref{eqn:cond_indep} and $Z \indep U \mid X$}\\
    &=\EEst{\EEst{M(Y^{a=0},\beta)}{A=1,U,X} \cdot \PPst{A=1}{Z=0, X}\cdot\frac{\PPst{Z=0}{X}}{\PPst{A=1}{X}}\cdot\frac{\PPst{A=1}{U,X}}{\PPst{Z=0}{X}}}{X}\\
    &= \EEst{\EEst{M(Y^{a=0},\beta)}{A=1,U,X} \cdot \frac{\PPst{A=1}{U,X}}{\PPst{A=1}{X}}}{X} \cdot \PPst{A=1}{Z=0, X}\\
    &= \EEst{\EEst{M(Y^{a=0},\beta)}{A=1,U,X} \cdot \frac{\One{A=1}}{\PPst{A=1}{X}}}{X} \cdot \PPst{A=1}{Z=0, X}\\
    &= \EEst{\EEst{M(Y^{a=0},\beta)}{A=1,U,X}}{A=1,X} \cdot \PPst{A=1}{Z=0, X}\\
    &=\EEst{M(Y^{a=0},\beta)}{A=1,X} \cdot \PPst{A=1}{Z=0, X}.
\end{align*}
Therefore, putting everything together, we obtain
\begin{align*}
    \delta^{M,A}(\beta,X) &= \EEst{M(Y^{a=0},\beta)}{A=1,X} \cdot \PPst{A=1}{Z=0, X} \cdot \left(1 -  \frac{\PPst{A=1}{Z=1,X}}{\PPst{A=1}{Z=0,X}}\right)\\
    &= -\EEst{M(Y^{a=0},\beta)}{A=1,X} \cdot \delta^A(X),
\end{align*}
and thus
\[-\EEst{\frac{\delta^{M,A}(\beta,X)}{\delta^A(X)}}{A=1} = \EEst{ \EEst{M(Y^{a=0},\beta)}{A=1,X}}{A=1} = \EEst{M(Y^{a=0},\beta)}{A=1}.\]

\subsection{Proof of Theorem~\ref{thm:eif}}

Under a parametric submodel $P_t=f(Y,A,Z,X;t)$, we compute the derivative of $h(\beta,P_t)$ as follows.

\begin{multline*}
    \pt h(\beta, P_t) \bigg\vert_{t=0} = \EE{\pt \delta_t^{M,A}(\beta,X)\cdot\frac{\Ppst{t}{A=1}{X}}{\delta_t^A(X)}} + \EE{\delta_t^{M,A}(\beta,X)\cdot\pt\left(\frac{\Ppst{t}{A=1}{X}}{\delta_t^A(X)}\right)}\\
    +\EE{\delta_t^{M,A}(\beta,X)\cdot\frac{\Ppst{t}{A=1}{X}}{\delta_t^A(X)}\cdot S(X;t)} \bigg\vert_{t=0}.
\end{multline*}
First, for $z \in \{0,1\}$, we compute
\begin{align*}
    &\pt \Epst{t}{M( Y,\beta) (1-A)}{Z=z,X} \bigg\vert_{t=0} = \EEst{M( Y,\beta) (1-A) \cdot S(Y,A \mid Z=z, X)}{Z=z,X}\\
    &= \EEst{\big(M( Y,\beta) (1-A) - \EEst{M( Y,\beta) (1-A)}{Z=z,X}\big) \cdot S(Y,A \mid Z=z, X)}{Z=z,X}.
\end{align*}
It follows that
\begin{align*}
    &\pt \delta_t^{M,A}(\beta,X) \bigg\vert_{t=0}\\
    &= \sum_{z=0,1} \EEst{\big(M( Y,\beta) (1-A) - \EEst{M( Y,\beta) (1-A)}{Z=z,X}\big) \cdot S(Y,A \mid Z=z, X)}{Z=z,X} \cdot (2z-1)\\
    &= \EEst{\big(M( Y,\beta) (1-A) - \EEst{M( Y,\beta) (1-A)}{Z,X}\big)\cdot\frac{2Z-1}{\pi_Z(X)}\cdot S(Y,A \mid Z, X)}{X}\\
    &=\EEst{\big(M( Y,\beta) (1-A) - \EEst{M( Y,\beta) (1-A)}{Z,X}\big)\cdot\frac{2Z-1}{\pi_Z(X)}\cdot S(Y,A,Z,X)}{X}.
\end{align*}
For the second term, we have
\begin{align*}
    \pt\left(\frac{\Ppst{t}{A=1}{X}}{\delta_t^A(X)}\right) = \frac{1}{\delta_t^A(X)}\pt \Ppst{t}{A=1}{X} - \frac{\Ppst{t}{A=1}{X}}{\delta_t^A(X)^2}\cdot \pt \delta_t^A(X).
\end{align*}
We then compute
\begin{multline*}
    \pt\Ppst{t}{A=1}{X}\bigg\vert_{t=0} = \EEst{A \cdot S(A \mid X)}{X} = \EEst{(A-\PPst{A=1}{X})\cdot S(A \mid X)}{X}\\
    = \EEst{(A-\PPst{A=1}{X})\cdot S(A, X)}{X} = \EEst{(A-\PPst{A=1}{X})\cdot S(Y,A,Z,X)}{X}.
\end{multline*}
Next,
\begin{multline*}
    \pt \Ppst{t}{A=1}{Z=z, X}\bigg\vert_{t=0} = \pt \EEst{A}{Z=z,X}\bigg\vert_{t=0} = \EEst{A \cdot S(A \mid Z=z, X)}{X}\\
    =\EEst{(A-\EEst{A}{Z=z,X}) \cdot S(A \mid Z=z, X)}{X} = \EEst{(A-\EEst{A}{Z=z,X}) \cdot S(A , Z=z \mid X)}{X}.
\end{multline*}
It follows that
\begin{equation}\label{eqn:dt_delta_A}
\begin{split}
    \pt \delta_t^A(X)\bigg\vert_{t=0} &= \pt \Ppst{t}{A=1}{Z=1, X} - \pt \Ppst{t}{A=1}{Z=0, X}\bigg\vert_{t=0}\\
    &=\EEst{(A - \PPst{A=1}{Z,X})\cdot\frac{2Z-1}{\pi_Z(X)}\cdot S(A,Z \mid X)}{X}\\
    &= \EEst{(A - \PPst{A=1}{Z,X})\cdot\frac{2Z-1}{\pi_Z(X)}\cdot S(Y,A,Z,X)}{X}.
\end{split}
\end{equation}
Therefore,
\begin{multline*}
    \pt\left(\frac{\Ppst{t}{A=1}{X}}{\delta^A(X)}\right) \bigg\vert_{t=0}\\
    = \EEst{\left(\frac{(A-\PPst{A=1}{X})}{\delta^A(X)}-(A - \PPst{A=1}{Z,X})\cdot\frac{\PPst{A=1}{X}}{\delta^A(X)^2}\cdot\frac{2Z-1}{\pi_Z(X)}\right)\cdot S(Y,A,Z,X)}{X}
\end{multline*}
Next, we can write the third term as
\begin{multline*}
\EE{\delta^{M,A}(\beta,X)\cdot\frac{\PPst{A=1}{X}}{\delta^A(X)}\cdot S(X)} = \EE{\left(\delta^{M,A}(\beta,X)\cdot\frac{\PPst{A=1}{X}}{\delta^A(X)}-h(\beta,P)\right)\cdot S(X)}\\
= \EE{\left(\delta^{M,A}(\beta,X)\cdot\frac{\PPst{A=1}{X}}{\delta^A(X)}-h(\beta,P)\right)\cdot S(Y,A,Z,X)}.
\end{multline*}
Putting everything together, the efficient influence function is given as
\begin{align*}
    &\dot{h}(\beta,O,P)\\
    = &\big(M( Y,\beta) (1-A) - \EEst{M( Y,\beta) (1-A)}{Z,X}\big)\cdot\frac{\PPst{A=1}{X}}{\delta^A(X)}\cdot\frac{2Z-1}{\pi_Z(X)}\\
    &+\delta^{M,A}(\beta,X) \cdot\left(\frac{A-\PPst{A=1}{X}}{\delta^A(X)}-(A - \PPst{A=1}{Z,X})\cdot\frac{\PPst{A=1}{X}}{\delta^A(X)^2}\cdot\frac{2Z-1}{\pi_Z(X)}\right)\\
    &+\delta^{M,A}(\beta,X)\cdot\frac{\PPst{A=1}{X}}{\delta^A(X)}-h(\beta,P)\\
    = &\Big(M( Y,\beta) (1-A) - \EEst{M( Y,\beta) (1-A)}{Z,X} - (A - \PPst{A=1}{Z,X})\cdot\delta(\beta,X)\Big)\cdot\frac{\PPst{A=1}{X}}{\delta^A(X)}\cdot\frac{2Z-1}{\pi_Z(X)}\\
    &+ (\delta(\beta,X)\cdot A - h(\beta,P)).
\end{align*}
Now we prove the second claim. The proof follows the steps similar to the arguments in~\citet{levis2024nonparametric}. 
We have
\begin{align*}
    &h(\beta,\bar{P}) - h(\beta,P) + \Ep{P}{\dot{h}(\beta,O,\bar{P})}\\
    &= h(\beta,\bar{P}) - \Ep{P}{\delta \cdot \rho}+ \Ep{P}{\bar{\delta}A - h(\beta,\bar{P}) + \frac{\bar{\rho}}{\bar{\delta}^A}\cdot \frac{2Z-1}{\bar{\pi}_Z} \cdot \left(M(Y,\beta)(1-A) - \bar{\mu}_Z - (A - \bar{\lambda}_Z)\cdot \bar{\delta}\right)}\\
    &=\Ep{P}{\bar{\delta}A - \delta \rho + \frac{\bar{\rho}}{\bar{\delta}^A}\cdot \frac{2Z-1}{\bar{\pi}_Z} \cdot \left(M(Y,\beta)(1-A) - \bar{\mu}_Z - (A - \bar{\lambda}_Z)\cdot \bar{\delta}\right)}\\
    &=\Ep{P}{(\bar{\delta} - \delta)\rho + \frac{\bar{\rho}}{\bar{\delta}^A}\cdot \frac{2Z-1}{\bar{\pi}_Z} \cdot \left(\mu_Z - \bar{\mu}_Z - (\lambda_Z - \bar{\lambda_Z})\cdot \bar{\delta}\right)} \qquad\textnormal{by conditioning on $(Z,X)$}\\
    &=\Ep{P}{(\bar{\delta} - \delta)\rho +\frac{\bar{\rho}}{\bar{\delta}^A}\cdot \frac{\pi_1}{\bar{\pi}_1}\cdot (\mu_1 - \bar{\mu}_1 - (\lambda_1 - \bar{\lambda}_1)\bar{\delta}) - \frac{\bar{\rho}}{\bar{\delta}^A}\cdot \frac{\pi_0}{\bar{\pi}_0}\cdot (\mu_0 - \bar{\mu}_0 - (\lambda_0 - \bar{\lambda}_0)\bar{\delta})} \quad\textnormal{by conditioning on $X$}\\
    &=\Ep{P}{(\bar{\delta} - \delta)(\rho-\bar{\rho}) + (\bar{\delta} - \delta)\bar{\rho}+\frac{\bar{\rho}}{\bar{\delta}^A}\cdot \frac{\pi_1}{\bar{\pi}_1}\cdot (\mu_1 - \bar{\mu}_1 - (\lambda_1 - \bar{\lambda}_1)\bar{\delta}) - \frac{\bar{\rho}}{\bar{\delta}^A}\cdot \frac{\pi_0}{\bar{\pi}_0}\cdot (\mu_0 - \bar{\mu}_0 - (\lambda_0 - \bar{\lambda}_0)\bar{\delta})}\\
    &=\Ep{P}{(\bar{\delta} - \delta)(\rho-\bar{\rho}) + \frac{\bar{\rho}}{\bar{\delta}^A} \cdot \tau},
\end{align*}
where
\[\tau = \bar{\delta}^A (\bar{\delta}-\delta) + \frac{\pi_1}{\bar{\pi}_1}\cdot (\mu_1 - \bar{\mu}_1 - (\lambda_1 - \bar{\lambda}_1)\bar{\delta}) - \frac{\pi_0}{\bar{\pi}_0}\cdot (\mu_0 - \bar{\mu}_0 - (\lambda_0 - \bar{\lambda}_0)\bar{\delta}).\]
The function $\tau$ can be simplified as follows. We first decompose $\bar{\delta}^A(\bar{\delta} - \delta)$ as
\begin{align*}
    \bar{\delta}^A(\bar{\delta} - \delta) = \bar{\delta}^{M,A} - \bar{\delta}^A \cdot \frac{\delta^{M,A}}{\delta^A} &= (\bar{\delta}^A - \delta^A)\left(\frac{\bar{\delta}^{M,A}}{\bar{\delta}^A} - \frac{\delta^{M,A}}{\delta^A}\right)+\bar{\delta}^{M,A} - \delta^{M,A} - \frac{\bar{\delta}^{M,A}}{\bar{\delta}^A}(\bar{\delta}^A - \delta^A)\\
    &= (\bar{\delta}^A - \delta^A)(\bar{\delta} - \delta)+(\bar{\mu}_1-\mu_1)-(\bar{\mu}_0-\mu_0)- \bar{\delta}((\bar{\lambda}_1 - \lambda_1) - (\bar{\lambda}_0 - \lambda_0)).
\end{align*}
It follows that
\begin{align*}
    \tau &= (\bar{\delta}^A - \delta^A)(\bar{\delta} - \delta) + (\bar{\mu}_1-\mu_1) \left(1-\frac{\pi_1}{\bar{\pi}_1}\right) - (\bar{\mu}_0-\mu_0)\left(1-\frac{\pi_0}{\bar{\pi}_0}\right) - \bar{\delta}(\bar{\lambda}_1 - \lambda_1)\left(1-\frac{\pi_1}{\bar{\pi}_1}\right) + \bar{\delta}(\bar{\lambda}_0 - \lambda_0)\left(1-\frac{\pi_0}{\bar{\pi}_0}\right)\\
    &= ((\bar{\lambda}_1 - \lambda_1) - (\bar{\lambda}_0 - \lambda_0))(\bar{\delta} - \delta) + \frac{1}{\bar{\pi}_1}\left(\bar{\mu}_1-\mu_1 - \bar{\delta}(\bar{\lambda}_1 - \lambda_1)\right)(\bar{\pi}_1-\pi_1)-\frac{1}{\bar{\pi}_0}\left(\bar{\mu}_0-\mu_0 - \bar{\delta}(\bar{\lambda}_0 - \lambda_0)\right)(\bar{\pi}_0-\pi_0)\\
    &= (\bar{\delta}^A - \delta^A)(\bar{\delta} - \delta) + (\bar{\pi}_1-\pi_1)\left(\frac{\bar{\mu}_1-\mu_1 - \bar{\delta}(\bar{\lambda}_1 - \lambda_1)}{\bar{\pi}_1} + \frac{\bar{\mu}_0-\mu_0 - \bar{\delta}(\bar{\lambda}_0 - \lambda_0)}{1-\bar{\pi}_1}\right).
\end{align*}
Putting everything together, we obtain
\begin{multline*}
    h(\beta,\bar{P}) - h(\beta,P) + \Ep{P}{\dot{h}(\beta,O,\bar{P})}  = \mathbb{E}_{P}\bigg[(\bar{\delta} - \delta)\left(\frac{\bar{\rho}}{\bar{\delta}^A}(\bar{\delta}^A - \delta^A) - (\bar{\rho}-\rho)\right)\\
    + \frac{\bar{\rho}}{\bar{\delta}^A}(\bar{\pi}_1-\pi_1)\left(\frac{\bar{\mu}_1-\mu_1 - \bar{\delta}(\bar{\lambda}_1 - \lambda_1)}{\bar{\pi}_1} + \frac{\bar{\mu}_0-\mu_0 - \bar{\delta}(\bar{\lambda}_0 - \lambda_0)}{1-\bar{\pi}_1}\right)\bigg],
\end{multline*}
which proves the claim.

\subsection{Proof of Corollary~\ref{cor:mult_robust}}
The first and the third condition follow directly from the formula of $R(\bar{P}, P)$ given in Theorem~\ref{thm:eif}. The second condition follows from the following observation:

\begin{align*}
    &(\bar{\mu}_1-\mu_1) - \bar{\delta}(\bar{\lambda}_1 - \lambda_1) = (\bar{\mu}_1-\mu_1) - \bar{\delta}(\bar{\lambda}_1 - \bar{\lambda}_0 + \bar{\lambda}_0 - \lambda_1) = (\bar{\mu}_1-\mu_1) - (\bar{\mu}_1-\bar{\mu}_0) + \bar{\delta}(\lambda_1-\bar{\lambda}_0)\\
    &= - \mu_1 + \bar{\mu}_0 + (\bar{\delta} - \delta)(\lambda_1-\bar{\lambda}_0) + \delta(\lambda_1-\lambda_0 + \lambda_0 - \bar{\lambda}_0) = - \mu_1 + \bar{\mu}_0 + (\bar{\delta} - \delta)(\lambda_1-\bar{\lambda}_0) + \mu_1 - \mu_0 + \delta( \lambda_0 - \bar{\lambda}_0)\\
    &= \bar{\mu}_0-\mu_0 + (\bar{\delta} - \delta)(\lambda_1-\bar{\lambda}_0) + \delta( \lambda_0 - \bar{\lambda}_0),
\end{align*}
which implies that $R(\bar{P},P) = 0$ holds if $\bar{\delta} = \delta$, $\bar{\mu}_0 = \mu_0$, and $\bar{\lambda}_0 = \lambda_0$.

\subsection{Proof of Corollary~\ref{cor:asymp}}

The proof applies the idea from~\citet{kennedy2016semiparametric}. We first show that the following decomposition holds.

\[\hat{h}(\beta) - h(\beta,P) = \mathbb{P}_n(\dot{h}(\beta,O,P)) + (\mathbb{P}_n - P) \big(\dot{h}(\beta,O,\hat{P})-\dot{h}(\beta,O,P)\big) + R(\hat{P},P).\]
Here, we abuse the notation $P$ to also represent the expectation with respect to the distribution $P$. Since $\Ep{P}{\dot{h}(\beta,O,P)} = 0$, the right hand side is equal to
\begin{align*}
    &\mathbb{P}_n(\dot{h}(\beta,O,P)) + \mathbb{P}_n \big(\dot{h}(\beta,O,\hat{P})-\dot{h}(\beta,O,P)\big) - \Ep{P}{\dot{h}(\beta,O,\hat{P})} + R(\hat{P},P)\\
    &=\mathbb{P}_n(\dot{h}(\beta,O,P)) + \mathbb{P}_n \big(\dot{h}(\beta,O,\hat{P})-\dot{h}(\beta,O,P)\big) + h(\beta,\hat{P}) - h(\beta,P) \qquad \textnormal{ by definition of $R(\hat{P},P)$}\\
    &= \mathbb{P}_n \big(\dot{h}(\beta,O,\hat{P})\big)+ h(\beta,\hat{P}) - h(\beta,P)\\
    &= \hat{h}(\beta) - h(\beta,P) \qquad \textnormal{ by definition of $\hat{h}(\beta)$},
\end{align*}
indicating that the decomposition holds. If $\|h(\beta,\hat{P}) - h(\beta,P)\| = o_P(1)$, then we have $(\mathbb{P}_n - P) \big(\dot{h}(\beta,O,\hat{P})-\dot{h}(\beta,O,P)\big) = o_P(n^{-1/2})$ by applying the following result.

\begin{lemma}[\citet{kennedy2020sharp}]
Let $\mathbb{P}_n$ denote the empirical measure over $(O_1,\cdots,O_n)$ and let $\hat{f}$ be a function, constructed independently of $(O_1,\cdots,O_n)$. Then
\[(\mathbb{P}_n - P)(\hat{f}-f) = O_P\left(\frac{\|\hat{f}-f\|}{\sqrt{n}}\right).\]
\end{lemma}

Therefore, under the additional assumption of $R(\hat{P},P) = o_P(n^{-1/2})$, we have
\[\sqrt{n}(\hat{h}(\beta) - h(\beta,P)) = \sqrt{n}\cdot \mathbb{P}_n(\dot{h}(\beta,O,P)) + o_P(1), \]
and the claim is proved by applying the central limit theorem along with Slutsky's theorem.

\subsection{Proof of Theorem~\ref{thm:unif_conv}}
We first provide a brief review of the result from~\citet{chernozhukov2018double}, which establishes a uniform convergence result for the cross-fitting estimator with linear scores. We then apply their results to our estimator to complete the proof of Theorem~\ref{thm:unif_conv}.

\subsubsection{Review of the results from~\texorpdfstring{\citet{chernozhukov2018double}}{1}}\label{sec:cher}

Suppose we aim to learn a parameter $\theta^*$, which is the solution to the equation $\Ep{P}{\psi(O, \theta, \eta^*)} = 0$,
where $\eta^* \in \Omega \subset \R^d$ denotes the true nuisance function and $\psi$ is a score function of the form $\psi(o, \theta, \eta) = \phi(o, \eta) - \theta$\footnote{\citet{chernozhukov2018double} provides results for general linear scores of the form $\psi(o, \theta, \eta) = \psi_1(o, \eta) \theta + \psi_2(o, \eta)$. However, for conciseness, we restrict our discussion here to this special case, which is particularly relevant to our setting.}

Given dataset $(O_i)_{1 \leq i \leq n}$, let us consider a cross-fitting estimator $\hat{\theta}$, constructed as follows:
\begin{enumerate}
    \item Construct a partition $(I_k)_{k \in [K]}$ of $[n]$, where each fold has size $|I_k| = n/K$.
    
    \item For $k=1,2,\cdots,K$, repeat the following:
    
    - Using the data $(O_i)_{i \in I_k^c}$, construct the nuisance estimators $\hat{\eta}_{-k}$.

    - Construct the estimator $\hat{\theta}_k$ as the solution of $\Ep{I_k}{\psi(O, \theta, \hat{\eta}_{-k})} = 0$, where $\mathbb{E}_{I_k}$ denotes the sample mean over the fold $I_k$.

    \item Compute the cross-fitting estimator $\hat{\theta} = \frac{1}{K}\sum_{k=1}^K \hat{\theta}_k$.

\end{enumerate}
We assume the following conditions for the score function, in addition to its linearity.

\begin{assumption}\label{asm:score}
    For $P \in \mathcal{P}_n$, the following conditions hold:
    \begin{enumerate}
        \item The map $\eta \mapsto \Ep{P}{\psi(O,\theta,\eta)}$ is twice continuously Gateaux-differentiable
        \item The score $\psi$ satisfies the Neyman orthogonality, i.e., $\partial_\eta\, \Ep{P}{\psi(O,\theta^*,\eta^*)} [\eta-\eta_0] = 0$, for all $\eta \in \Omega$.
    \end{enumerate}
\end{assumption}

Next, we assume the following for the quality of nuisance estimators.

\begin{assumption}\label{asm:nuisance_reg}
 There exists a constant $c > 0$ and sequences $(P_n)_{n \in \mathbb{N}} \subset P$, $(\tau_n)_{n \in \mathbb{N}} \subset \mathbb{R}^+$, and $(\eps_n)_{n \in \mathbb{N}} \subset \mathbb{R}^+$ such that $\lim_{n \to \infty} \tau_n = \lim_{n \to \infty} \eps_n = 0$, and the following conditions hold: For all $P \in \mathcal{P}_n$ and each $k \in [K]$, the nuisance estimator $\hat{\eta}_{-k}$ belongs to a set $\Omega_n$ with probability at least $1 - \eps_n$, where $\Omega_n$ contains $\eta^*$ and satisfies:
\begin{align*}
    &\sup_{\eta \in \Omega_n} \Ep{P}{|\psi(O, \theta^*, \eta)|^2}^{1/2} \leq c, \\
    &\sup_{\eta \in \Omega_n} \Ep{P}{|\psi(O, \theta^*, \eta) - \psi(O, \theta^*, \eta^*)|^2}^{1/2} \leq \tau_n, \\
    &\sup_{\eta \in \Omega_n} \Big|\Ep{P}{\psi(O, \theta^*, \eta)} - \Ep{P}{\psi(O, \theta^*, \eta^*)} + \partial_r\Ep{P}{\psi(O, \theta^*, \eta^* + r(\eta - \eta^*))}\big\vert_{r=0} \Big|  \leq \frac{\tau_n}{\sqrt{n}}
\end{align*}
We note that~\citet{chernozhukov2018double} impose the following stronger condition in place of the third one:
\[\sup_{\substack{\eta \in \Omega_n \\ r \in (0, 1)}} \left| \partial_r^2 \Ep{P}{\psi(O, \theta^*, \eta^* + r(\eta - \eta^*))} \right| \leq \frac{\tau_n}{\sqrt{n}}.\]
The third condition in Assumption~\ref{asm:nuisance_reg} follows from this stronger condition via Taylor's theorem. In fact, what is essentially used in their proof is the weaker condition stated in Assumption~\ref{asm:nuisance_reg}---they simply apply Taylor’s theorem within their argument to reduce to it.

In addition, the variance of the score at the true nuisance function is non-degenerate, i.e.,
\[\Ep{P}{\psi(O, \theta^*, \eta^*) \psi(O, \theta^*, \eta^*)^\top} \succeq c' I\]
for a constant $c' > 0$.

\end{assumption}

Under the above conditions, the following lemma shows that the confidence interval from the normal approximation is asymptotically valid, uniformly over the neighborhood of the true nuisance function.

\begin{lemma}[\citet{chernozhukov2018double}]\label{lemu:unif_conv}
    Under Assumptions~\ref{asm:score} and~\ref{asm:nuisance_reg}, for a fixed vector $w \in \R^d$, the confidence interval
    \[\widehat{C}^w = \left[w^\top \hat{\theta} - z_{\alpha/2}\cdot \sqrt{\frac{w^\top \hat{\sigma}^2 w}{n}}, w^\top \hat{\theta} + z_{\alpha/2}\cdot \sqrt{\frac{w^\top \hat{\sigma}^2 w}{n}}\right]\]
    satisfies
    \[\lim_{n \rightarrow \infty} \sup_{P \in \mathcal{P}_n} \left|\Pp{P}{w^\top \theta^* \in \widehat{C}^w} - (1-\alpha)\right| = 0.\]
\end{lemma}

\subsubsection{Proof of the theorem}

Now, we prove Theorem~\ref{thm:unif_conv} by applying the results above. Observe that the steps for constructing our estimator and the confidence interval correspond to those in Section~\ref{sec:cher}, with the score function
\begin{align*}
    \psi(o, \theta, \eta) &= \delta(\beta,x)\cdot a + \frac{\rho(x)}{\delta^A(x)}\cdot\frac{2z-1}{\pi_z(x)}
    \cdot\Big(M(y,\beta) (1-a) - \mu_z(\beta,x) - (a - \lambda_z(x))\cdot\delta(\beta,x)\Big)-\theta\\
    &= \frac{\mu_1 - \mu_0}{\lambda_1 - \lambda_0}\cdot a + \frac{\rho}{\lambda_1 - \lambda_0}\cdot\left(\frac{z}{\pi_1} - \frac{1-z}{1-\pi_1}\right)\cdot\bigg(M(y,\beta) (1-a) - \mu_1 z - \mu_0 (1-z)\\
    &\hspace{85mm} -(a - \lambda_1 z - \lambda_0 (1-z))\cdot\frac{\mu_1 - \mu_0}{\lambda_1 - \lambda_0}\bigg) - \theta,
\end{align*}
where $\eta = (\rho,\pi_1,\mu_0,\mu_1, \lambda_0, \lambda_1)$, and we omit the variables $x$ and $\beta$ for simplicity. Therefore, it remains to verify that Assumptions~\ref{asm:score} and~\ref{asm:nuisance_reg} holds. 

Since Assumption~\ref{asm:score} holds clearly for the above efficient influence function-based score $\psi$, it is sufficient to check the conditions in Assumption~\ref{asm:nuisance_reg}, to apply Lemma~\ref{lemu:unif_conv}. First, the non-degeneracy follows directly from the second condition in Assumption~\ref{asm:regularity}. Next, let us define the space $\Omega_n$ as
\begin{multline*}
\Omega_n = \Big\{ \eta \in \Omega : \|\eta(X)\|_{\infty} \leq c_1 \text{ a.s., } \Ep{P}{\|\eta (X) - \eta^*(X)\|^2}^{1/2} \leq c_2 \cdot \tau_n, \\
\text{and the product biases of } \eta \text{ are bounded by } c_3 \cdot \tau_n / \sqrt{n} \Big\}.,
\end{multline*}
where the ``product biases" above refers to 
\begin{alignat*}{3}
    &\|(\delta - \delta^*)({\delta^A} - {\delta^A}^*)\|,\quad
    &&\|(\delta - \delta^*)(\rho - \rho^*)\|,\quad
    &&\|(\pi_1 - \pi_1^*)(\mu_1 - \mu_1^*)\|,\\
    &\|(\pi_1 - \pi_1^*)(\mu_0 - \mu_0^*)\|,\quad
    &&\|(\pi_1 - \pi_1^*)(\lambda_1 - \lambda_1^*)\|,\quad
    &&\|(\pi_1 - \pi_1^*)(\lambda_0 - \lambda_0^*)\|.
\end{alignat*}
By the first condition in Assumption~\ref{asm:regularity}, there exists a sequence $(\eps_n)_{n \in \N}$ converging to zero such that the nuisance estimator $\hat{\eta}_{-k}$ belongs to $\Omega_n$ with probability at least $1 - \eps_n$. Therefore, it remains to show that the uniform-bound conditions in Assumption~\ref{asm:nuisance_reg} hold. 

The first condition trivially holds from the condition $\|\eta(X)\|_\infty \leq c_1$ of $\Omega_n$, i.e., all the nuisance function components are bounded. The second condition follows from the two conditions $\|\eta(X)\|_\infty \leq c_1$ and $\Ep{P}{\|\eta (X) - \eta^*(X)\|^2}^{1/2} \leq c_2 \cdot \tau_n$. The third condition---the uniform bound on the second-order remainder---follows directly from the expression in Theorem~\ref{thm:eif}, along with the constraints defining $\Omega_n$, which bound the product biases and $|\eta(X)|_\infty$, and the third part of Assumption~\ref{asm:regularity}. Therefore, applying Lemma~\ref{lemu:unif_conv} completes the proof.


\subsection{Proof of Corollary~\ref{cor:CI_beta}}
For intuition, we first present the proof in the one-dimensional case, where the confidence interval is essentially given as~\eqref{eqn:CI_beta}. Recall that $h(\beta^*,P) = 0$ by the definition of $\beta^*$. Thus, we have
\begin{align*}
    \PP{\beta^* \in \chb} = \PP{0 \in \chh(\beta^*)} = \PP{h(\beta^*,P) \in \chh(\beta^*)},
\end{align*}
and the claim follows by applying Corollary~\ref{cor:moment_cov}.

The proof for the general case relies on the asymptotic normality result
\[\sqrt{n} \cdot \hat{\Sigma}^{-1/2} (\hat{h}(\beta) - h(\beta,P)) \stackrel{D}{\rightarrow} \mathcal{N}(0,I_d),\]
which is established in~\citet{chernozhukov2018double} and is also used in the proof of Theorem~\ref{thm:unif_conv}. This implies that at $\beta = \beta^*$, we have
\[n \|\hat{\Sigma}^{-1/2} \hat{h}(\beta^*)\|^2 \stackrel{D}{\rightarrow} \chi_d^2\]
by continuous mapping theorem, and the desired claim follows by the same argument as above.

\subsection{Proof of Corollary~\ref{cor:qtt}}
The proof applies the idea of~\citet{berger1994p}. Note that it is sufficient to prove that $\lim_{n \rightarrow \infty} \PP{p(\xi^*) > \alpha} \geq 1-\alpha$, since $p(\xi^*) > \alpha$ implies $\xi^* \in \widehat{C}_{\xi^*}$. Then compute
\begin{multline*}
    \PP{p(\xi^*) \leq \alpha} = \PP{p(\xi^*) < \alpha, \gamma^* \in \widehat{C}_{\gamma^*}} +\PP{p(\xi^*) \leq \alpha, \gamma^* \notin \widehat{C}_{\gamma^*}}\\
    \leq \PP{p(\xi^*;\gamma^*) + \zeta \leq \alpha, \gamma^* \notin \widehat{C}_{\gamma^*}} +\PP{\gamma^* \notin \widehat{C}_{\gamma^*}} \leq \PP{p(\xi^*;\gamma^*) \leq \alpha - \zeta} + \zeta.
\end{multline*}
Next, applying Corollary~\ref{cor:CI_beta}, we have $\lim_{n \rightarrow \infty} \PP{p(\xi^*;\gamma^*) \leq \alpha - \zeta} = \alpha - \zeta$. Putting these results together, we obtain
\[\lim_{n \rightarrow \infty} \PP{p(\xi^*) \leq \alpha} \leq \alpha,\]
which implies the desired inequality.

\subsection{Proof of Proposition~\ref{prop:qte}}

Observe that
\begin{align*}
    &\EE{M(Y,X,Z,A,\beta^*)} = \EE{\left(\One{Y \leq \gamma^*+\beta^*(1-A)} - q\right)\cdot\left(\frac{2Z-1}{\pi_Z(X)}\cdot\frac{\rho(X)}{\delta^A(X)}\right)^{1-A}}\\
    &= \EE{\left(\One{Y \leq \gamma^*+\beta^*(1-A)} - q\right) \cdot A + \left(\One{Y \leq \gamma^*+\beta^*(1-A)} - q\right) \cdot \left(\frac{2Z-1}{\pi_Z(X)}\cdot\frac{\rho(X)}{\delta^A(X)}\right) \cdot (1-A)}\\
    &= \EEst{\One{Y \leq \gamma^*} - q}{A=1}\cdot\PP{A=1}\\
    &\hspace{48mm}+ \EEst{\left(\One{Y \leq \gamma^*+\beta^*} - q\right) \cdot \left(\frac{2Z-1}{\pi_Z(X)}\cdot\frac{\rho(X)}{\delta^A(X)}\right)}{A=0}\cdot\PP{A=0},
\end{align*}
and thus it is sufficient to show that
\begin{equation}\label{eqn:qte_eq}
    \EEst{\left(\One{Y \leq \gamma^*+\beta^*} - q\right) \cdot \left(\frac{2Z-1}{\pi_Z(X)}\cdot\frac{\rho(X)}{\delta^A(X)}\right)}{A=0} = 0
\end{equation}
holds, since $\EEst{\One{Y \leq \gamma^*} - q}{A=1} = \EEst{\One{Y^{a=1} \leq \gamma^*} - q}{A=1} = 0$ by definition of $\gamma^*$.

Now, by definition of $\gamma^*+\beta^* = F_{Y^{a=0} \mid A=1}^{-1}(q)$ and Theorem~\ref{thm:identification}, we have
\begin{align*}
    0 = \EEst{\One{Y^{a=0} \leq \gamma^* + \beta^*}-q}{A=1} = -\EE{\frac{\delta^{\tilde{M},A}(\beta^*,X)}{\delta^A(X)} \cdot \frac{\PPst{A=1}{X}}{\PP{A=1}}},
\end{align*}
where $\delta^{\tilde{M},A}(\beta,X)$ is defined as~\eqref{eqn:nui_delta} with $\tilde{M}(Y,\beta) = \One{Y \leq \gamma^*+\beta}-q$.
Therefore,
\begin{align*}
    &0 = \EE{\frac{\delta^{\tilde{M},A}(\beta^*,X)}{\delta^A(X)} \cdot \PPst{A=1}{X}} = \EE{\frac{\rho(X)}{\delta^A(X)}\cdot \sum_{z=0,1} \EEst{(\One{Y \leq \gamma^* + \beta}-q) (1-A)}{Z=z,X} (2z-1)}\\
    &= \EE{\frac{\rho(X)}{\delta^A(X)} \cdot \EEst{\EEst{(\One{Y \leq \gamma^* + \beta}-q) (1-A)}{Z,X} \cdot \frac{2Z-1}{\pi_Z(X)}}{X}}\\
    &= \EE{(\One{Y \leq \gamma^* + \beta}-q) (1-A) \cdot \frac{2Z-1}{\pi_Z(X)} \cdot \frac{\rho(X)}{\delta^A(X)}}\\
    &= \EEst{\left(\One{Y \leq \gamma^*+\beta^*} - q\right) \cdot \left(\frac{2Z-1}{\pi_Z(X)}\cdot\frac{\rho(X)}{\delta^A(X)}\right)}{A=0} \cdot \PP{A=0},
\end{align*}
and thus the equality in~\eqref{eqn:qte_eq} follows as desired.

\end{document}